  \documentclass[journal ,onecolumn]{IEEEtran}

\usepackage{graphicx,color}
\usepackage{xargs}
\usepackage[colorinlistoftodos,prependcaption]{todonotes}

\def\final{1}

\usepackage{amsthm}
\usepackage{amssymb}
\usepackage{amsmath}
\usepackage{setspace}                           %for doublespacing
\newtheorem{definitionenv}{Definition}
\newtheorem{lemmaenv}[definitionenv]{Lemma}
\newtheorem{theoremenv}[definitionenv]{Theorem}
\newtheorem{corollaryenv}[definitionenv]{Corollary}
\newtheorem{propositionenv}[definitionenv]{Proposition}
\newtheorem{conjectureenv}[definitionenv]{Conjecture}
\newtheorem{exampleenv}[definitionenv]{Example}
\newtheorem{app-lemmaenv}[section]{Lemma}

\newenvironment{lemma}{\begin{lemmaenv}\rm}{\end{lemmaenv}}
\newenvironment{theorem}{\begin{theoremenv}\rm}{\end{theoremenv}}

\newenvironment{example}{\begin{exampleenv}\rm}{\end{exampleenv}}

\newenvironment{conjecture}{\begin{conjectureenv}\rm}{\end{conjectureenv}}
\newenvironment{app-lemma}{\begin{app-lemmaenv}\rm}{\end{app-lemmaenv}}

\newcommand{\ba}{{\bf a}}
\newcommand{\bb}{{\bf b}}
\newcommand{\bc}{{\bf c}}

\newcommand{\e}{{\bf e}}
\newcommand{\f}{{\bf f}}
\newcommand{\bg}{{\bf g}}
\newcommand{\bh}{{\bf h}}

\newcommand{\bs}{{\bf s}}

\newcommand{\bu}{{\bf u}}
\newcommand{\bv}{{\bf v}}
\newcommand{\bw}{{\bf w}}
\newcommand{\x}{{\bf x}} 
\newcommand{\by}{{\bf y}}
\newcommand{\bz}{{\bf z}}

\newcommand{\cA}{{\cal A}}

\newcommand{\cC}{{\cal C}}

\newcommand{\cE}{{\cal E}}
\newcommand{\cF}{{\cal F}}
\newcommand{\cG}{{\cal G}}

\newcommand{\cN}{{\cal N}}

\newcommand{\cS}{{\cal S}}

\newcommand{\mF}{{\mathbb F}}

\newcommand{\be}{\begin{equation}}
\newcommand{\ee}{\end{equation}}

\newcommand{\bea}{\begin{eqnarray*}}
\newcommand{\eea}{\end{eqnarray*}}

\newcommand\wt{\mbox{{\rm wt}\,}}
\newcommand\Tr{\mbox{{\rm Tr}\,}}

\newcommand\dist{\mbox{{\rm dist}\,}}

\newcommand{\remove}[1]{}

\renewcommand{\le}{\leqslant}\renewcommand{\ge}{\geqslant}

\newcommand{\ket}[1]{|#1\rangle}

\newcommand{\ds}{_{\mathrm{DS}}}
\newcommand{\tD}{_{\mathrm{D}}}
\newcommand{\tS}{_{\mathrm{S}}}
\newcommand{\dsp}{_{\mathrm{DS}}^\perp}

\newcommand{\QED}{\qed}

\ifnum\final=0
\newcommand{\mynote}[2]{{\color{#1} \marginpar{\tiny #2}}}
\newcommand{\mybignote}[2]{{\color{#1} $\langle \langle$ #2$\rangle \rangle$}}
\newcommandx{\rednote}[2][1=]{\todo[linecolor=red,backgroundcolor=red!25,bordercolor=red,#1]{#2}}
\newcommandx{\bluenote}[2][1=]{\todo[linecolor=blue,backgroundcolor=blue!25,bordercolor=blue,#1]{#2}}
\newcommandx{\yellownote}[2][1=]{\todo[linecolor=yellow,backgroundcolor=yellow!25,bordercolor=yellow,#1]{#2}}
\newcommandx{\greennote}[2][1=]{\todo[inline,linecolor=olive,backgroundcolor=green!25,bordercolor=olive,#1]{#2}}

\else
\newcommand{\mynote}[2]{}
\newcommand{\mybignote}[2]{}
\newcommand{\rednote}[2][1=]{}
\newcommand{\bluenote}[2][1=]{}
\newcommand{\greennote}[2][1=]{}
\newcommand{\yellownote}[2][1=]{}

\fi

\newcommand{\cnote}[1]{\greennote[inline]{Ching-Yi: {#1}}}
\newcommand{\anote}[1]{\yellownote[inline]{Alexei: {#1}}}

\begin{document}
\title{Quantum  Data-Syndrome Codes}

%\author{\IEEEauthorblockN{Alexei Ashikhmin,}
%\IEEEauthorblockA{\textit{Nokia Bell Labs,} \\Murray Hill, NJ, USA\\
%alexei.ashikhmin@nokia-bell-labs.com}\\
%\and
%\IEEEauthorblockN{Ching-Yi Lai,}
%\IEEEauthorblockA{\textit{Institute of Communications, National Chiao Tung University,}\\
% Hsinchu 30010, Taiwan\\
%cylai@nctu.edu.tw}\\
%\and
%\IEEEauthorblockN{Todd A. Brun,}
%\IEEEauthorblockA{\textit{University of Southern California}\\
%  Los Angeles, CA, USA\\
%  tbrun@usc.edu}
%}

\author{Alexei Ashikhmin, 	Ching-Yi Lai, 	and
	Todd A. Brun
	\thanks{
		This work
		was presented in part at ISIT 2014 and in part at ISIT 2016.
		
		AA is with the Nokia Bell Labs, 600 Mountain Ave,
		Murray Hill, NJ 07974.		 (email: alexei.ashikhmin@nokia-bell-labs.com)
		
		CYL is with the Institute of Communications, National Chiao Tung University,
		Hsinchu 30010,  Taiwan.  		(email: cylai@nctu.edu.tw)
		
TAB is with the Electrical Engineering Department, University of Southern California,
Los Angeles, CA 90089, USA. (email: tbrun@usc.edu) 
}	
}

%\author{\IEEEauthorblockN{Alexei Ashikhmin}
%\IEEEauthorblockA{Bell Laboratories\\
%Alcatel-Lucent, 600 Mountain Ave\\
%Murray Hill, NJ 07974\\
%aea@reseach.bell-labs.com} \and \IEEEauthorblockN{Ching-Yi Lai}
%\IEEEauthorblockA{
%Institute of Information Science\\
%Academia Sinica\\
%%Room 409 Old Building
%No 128, Academia Road, Section 2\\
%Nankang, Taipei 11529, Taiwan\\
%cylai0616@iis.sinica.edu.tw}\and
%\IEEEauthorblockN{Todd A. Brun} \IEEEauthorblockA{Communication
%Sciences Institute\\
% University of Southern California\\
%  Los Angeles, California, USA\\
%  tbrun@usc.edu}}
%
%
%
%

%\author{Alexei Ashikhmin$^1$, Ching-Yi Lai

\maketitle

%I don't know why I have to reset thispagestyle, but otherwise get page numbers

\thispagestyle{empty}

\begin{abstract}
Performing active quantum error correction to protect fragile quantum states highly depends on the correctness of error information--error syndromes.
To obtain reliable error syndromes using imperfect physical circuits,
we propose the idea of quantum data-syndrome (DS) codes that are capable of correcting both data qubits and syndrome bits errors.
We study fundamental properties of quantum DS codes, including split weight enumerators, generalized MacWilliams identities, 
 and linear programming bounds. In particular, we derive Singleton and Hamming-type upper bounds on degenerate quantum DS codes.
 Then we study random DS codes and show that random DS codes with a relatively small additional syndrome measurements achieve the Gilbert-Varshamov bound of stabilizer codes. Constructions of quantum DS codes are also discussed. A family of quantum DS codes is based on classical linear block codes, called syndrome measurement codes,  so that syndrome bits are encoded in additional redundant stabilizer measurements.
 Another family of quantum DS codes is  CSS-type quantum DS codes based on classical cyclic codes,  and this includes the Steane code and the quantum Golay code.
 %, which are known to have fault-tolerant error thresholds.

\end{abstract}

\section{Introduction}

%An important issue in the implementation of a quantum computer is to protect quantum information from decoherence.
Quantum error-correcting codes provide a method of actively protecting quantum information~\cite{LB13}.
In a quantum error-correcting code, quantum information is
stored in the joint $+1$ eigenspace of a set of Pauli operators,
called \emph{stabilizers}.
To perform quantum error correction, we have to learn knowledge of errors, the \emph{error syndromes} (in bits), through quantum measurements.
More precisely, the error syndrome are given by measuring a generating set of the stabilizers.
Realistically, the quantum gates used to perform quantum
error correction are themselves faulty and thus the measurement outcomes
for the error syndrome can be wrong due to faulty
measurements or newly introduced errors from faulty gates.

In this article, we are interested in eliminating the effect
of faulty syndrome measurement.
Typically, this can be done with the syndromes being measured repeatedly
in the case of Shor’s syndrome
extraction~\cite{Shor96}. In other words, more redundant measurements than necessary are required to
determine a reliable error syndrome.
Other protocols have also been proposed to handle measurement errors for color codes and other topological codes~s~\cite{Bom15,BNB16,BDMT17}, following~Bombin's seminal work on the so-called single-shot fault-tolerant quantum error correction on color codes~\cite{Bom15}. Very recently Campbel has proposed a theory for this one-shot error correction~\cite{Cam19}.
%More generally, it is not necessary to determine an accurate error syndrome as long as the size of residual errors after quantum error correction is still within control.
Herein we consider a general scheme of quantum stabilizer codes that are capable of correcting data qubit errors and syndrome bit errors simultaneously with the help of additional stabilizer measurements.
These codes are called {\em quantum data-syndrome (DS)  codes}~\cite{ALB14, ALB16}.
%The approach in \cite{ALB14} can be combined in this new scheme.
%It is done by considering quantum error-correcting codes that can {\emph{correct}} data qubit errors and \emph{detect/correct} syndrome bit errors.
This idea is also independently studied by Fujiwara in \cite{Fuji14}.
Constructions and  simulations of LDPC and Convolutional DS codes have been studied in
\cite{ALB14,ZAWP19}. 
%More work on constructions of DS codes is needed.

\cnote{Please check this literature review. I am not confident with these references.}
%\cnote{to be continued}

In this paper, we give a comprehensive study of quantum DS codes, which completes our previous work~\cite{ALB14,ALB16}.
We first define quantum DS codes.
In addition to quantum data errors, a measurement outcome suffers a measurement error depending on the weight of the measured stabilizer. One may use redundant stabilizer measurements to decode the error syndrome first and then do quantum error correction. The overall procedure can be seen as decoding a larger code that is a concatenation of a binary code (for syndrome bits) and a quaternary code (for qubits). Thus we call such scheme a quantum DS code.
Given an $[[n,k,d]]$ stabilizer code that encodes $k$ information qubits in $n$ physical qubits with minimum distance $d$, we denote the corresponding DS code with additiona $r$ redundant stabilizer measurements by the parameters  $[[n,k,d:r]]$.
To get familiar with how quantum DS codes work,  we introduce a family of quantum DS codes
such that $r$ additional stabilizer measurements are based on classical error-correcting codes.  The idea of repeated syndrome measurements  is similar to using a classical repetition code.
We generalize this approach by introducing the idea of syndrome measurement
(SM) codes based on classical linear block codes. 
%Using SM codes, reliable syndromes will be decoded first, and then follows the quantum error correction.
Examples show that syndrome decoding can be improved using SM codes than repeted  syndrome measurements.

Then we define  notions of minimum distance and split weight enumerators.
Naturally, the generalized MacWilliams  identities hold for DS codes, which lead to the linear programming bounds on the minimum distance of small  DS codes. 
%This further justifies Fujiwara's example that the seven-qubit  code  is capable of correcting a single-qubit or a measurement error~\cite{Fuji14}.
%  It is not difficult to see that upper bounds on the underlying stabilizer codes are also upper bounds on the corresponding DS codes.
 Armed with the MacWilliams  identities for DS codes,
we further analyze algebraic linear programming bounds for DS codes, generalizing  the approach proposed in \cite{AL99}.
In the case of $r=0$, we derive Singleton and Hamming type upper bounds on the code size of degenerate quantum DS codes. Especially, we demonstrate that the Hamming bounds for nondegnerate codes and for degenerate codes will merge  for sufficiently large $n$.

Next we study the properties of random quantum DS codes for the case of $0\leq r\leq n-k$.
We will   show that the minimum distance of random DS codes with a
relatively small $r\leq n-k$ achieves the Gilbert-Varshamov bound of stabilizer codes.
 Along the way average weight enumerators are also derived, which may be of independent interest, since in classical coding theory, those numerators also lead to upper bound on the probability of decoding error.

 Finally, we consider  CSS-type (\cite{CS96,Ste96a}) quantum DS codes and provide DS code constructions from CSS-type quantum cyclic codes such that their minimum distances are as high as their underlying stabilizer codes.
 The quantum Golay code and Steane code, which are shown to have fault-tolerant error thresholds~\cite{PR12,AGP06}, are included
as two examples.

%A key point  to  cope with this problem is to design a quantum error-correcting code so that correctable data qubit errors and syndrome-bit errors have distinct error syndromes.
%We can have lower decoding complexity and fewer measurements, but slightly worse error-correcting ability.

%This paper is organized as follows. Preliminaries of Pauli operators and quantum stabilizer codes are given in the next section.
%We define quantum DS codes in Section~\ref{sec:quantum_data-syndrome}. SM codes are discussed in~Section~\ref{sec:SMcodes}. Next we define minimum distance and split weight enumerators for DS codes in Section~\ref{sec:swe}. Upper bounds on DS codes are derived in Section~\ref{sec:bounds}.
%Then we discuss random DS codes in Section~\ref{sec:random}.
%Finally we consider CSS-type DS codes in Section~\ref{sec:CSS}.
%%Then we conclude.

\section{Preliminaries}
%\cnote{I added this section to explain  in more detail the preliminaries}
%\subsection{Pauli operators and the Galios field $\mathbb{F}_4$}
In this paper we focus on two-dimenstional quantum systems: qubits.
An $n$-qubit state space is a $2^n$-dimensional complex Hilbert space $\mathbb{C}^{2^n}$ and a pure quantum state is a unit vector in the state space.
%The evolution of a closed quantum system can be described by a untary operator.
A basis of  the linear operators on the $n$-qubit state space is the $n$-fold Pauli operators
 $\{M_1\otimes \cdots \otimes M_n: M_j\in\{I,X,Y,Z\} \},$
where
$I=\begin{bmatrix}1 &0\\0&1\end{bmatrix}, X=\begin{bmatrix}0 &1\\1&0\end{bmatrix},  Z=\begin{bmatrix}1 &0\\0&-1\end{bmatrix}, Y=iXZ$
are the Pauli matrices. The $n$-fold Pauli group $\cG_n$ is the set of $n$-fold Pauli operators with possible phases $\pm 1, \pm i$.
%Let  $\mathbb{I}$ denote the identity matrix of appropriate dimensions.
Note that Pauli operators either commute or anticommute with each other.
We can define an inner product in ${\cG}_n$: for $g,h\in\cG_n$,
\begin{align}
\langle g,h\rangle_{{\cG}_n}=\begin{cases}
0, & \mbox{if  $gh=hg$};\\
 1, & \mbox{otherwise}.
\end{cases}
\end{align}
%\rmark{More details on Pauli operators and their properties can be found in \cite{NC00}.}

Often it is convenient to consider the quantum coding problem via codes over the Galios filed of four elements $\mF_4=\{0,1,\omega,{\omega^2}\}$~\cite{CRSS98}. We can define a homomorphism $\tau$ on $\cG_1$ that maps $I,X,Z,Y$ to
$0,1,\omega,\omega^2$, respectively, regardless of a possible phase $\pm 1,\pm i$ in front of a Pauli matrix, and this homomorphism extends to an $n$-fold Pauli operator naturally.
For example, $\tau(\pm i X\otimes Y\otimes Z\otimes I\otimes I)= 1\omega^2\omega00$.
 Throughout the text for an $n$-fold Pauli operator $g$,  we denote by $\bg\in \mF_4^{n}$ the corresponding vector and vice versa (up to an appropriate phase).
Likewise, we define a trace inner product on $\mathbb{F}_4^n$: for $\x=(x_1,x_2,\dots, x_n), \by=(y_1,y_2,\dots, y_n)\in \mathbb{F}_4^n$,
%\vspace{-0.2cm}
\begin{align}
\x*\by=\Tr^{\mF_4}_{\mF_2}\left(\sum_{i=0}^n x_i\overline{y}_i\right),  \label{eq:trace_ip}
\end{align}
where $\overline{y}_i$ denotes the conjugation of $y_i$ in $\mF_4=\{0,1,\omega,{\omega^2}\}$ with $\bar{0}=0,\bar{1}=1,
\bar{\omega}=\omega^2,$ and $\bar{\omega^2}=\omega$.
It can be checked that $\langle g,h\rangle_{{\cG}_n} = \bg * \bh$ for $g,h\in\cG_n$ and $\bg=\tau(g),\bh=\tau(h)$.

 %\subsection{Quantum Stabilizer Codes}

Suppose $\cS=\langle g_1, \dots, g_{n-k}\rangle$ is an Abelian subgroup of ${\cG}_n$, where $g_j$ are independent generators of $\cS$, such that the minus identity $-I^{\otimes n}\notin \cS$.
Then $\cS$ defines a quantum \emph{stabilizer} code %$\cC(\cS)$
$Q=\{\ket{\psi}\in \mathbb{C}^{2^n}: g\ket{\psi}=\ket{\psi}, \forall g\in \cS\}$
of dimension $2^k$.
The vectors $\ket{\psi}\in \cC(\cS)$ are called the \emph{codewords} of $\cC(\cS)$ and the operators $g\in\cS$  are called the \emph{stabilizers} of $\cC(\cS)$.
By the quantum error correction conditions~\cite{BDSW96,KL97}, it suffices to consider error correction on a discrete set of error operators. Thus we only treat errors that are Pauli operators in this paper.
%Suppose a Pauli error $e\in \cG_n$ occurs on a codeword $\ket{\psi}$.
%If $e$ anticommutes with some stabilizers $g_j$'s, it can be detected by measuring the observables $g_j$'s on $e\ket{\psi}$. On the other hand, if $e$ commutes with all the stabilizers, it cannot be detected.

%Let $Q$ be an  $[[n,k]]$ stabilizer code defined by a stabilizer group $\mathcal{S}$ with
% independent generators  $g_1,g_2,\dots,g_{n-k}\in\cG_n$ and corresponding vectors $\bg_1,\ldots,\bg_{n-k}\in \mF_4^n$.
In the following we will use the corresponding codes over $\mathbb{F}_4$ to discuss the quantum error correction procedure of $Q$.
Define a check matrix
\begin{align}
H=\left[\begin{array}{c}
\bg_1\\
\vdots\\
\bg_{n-k}\end{array}\right],
\end{align}
where $\bg_1,\ldots,\bg_{n-k}\in \mF_4^n$ are the  corresponding vectors of  $g_1,g_2,\dots,g_{n-k}\in\cG_n$.
Let $C$ be the classical $[n,n-k]$ code over $\mF_4$ generated  by the rows of $H$, and $C^\perp$ be its dual with respect
to the trace inner product $*$ defined in (\ref{eq:trace_ip}).
We have $C\subseteq C^{\perp}$ since $\bg_i *\bg_j=0$ for all $i,j$.
Suppose a quantum codeword $|\psi \rangle\in Q$
is corrupted by a Pauli error $e\in\cG_n$ and let $\e=\tau(e)\in \mF_4^n$  be the corresponding vector. Then the syndrome of $e$ is
$\bs=(s_1,\dots, s_{n-k})\in \mF_2^{n-k}$, where
$
s_j=\bg_j*\e.
$
% where $\bs\in \mF_2^{n-k}$.
More explicitly, the syndrome $\bs$ has the commutation relations between the Pauli error $e$ and the stabilizers $g_1,g_2,\dots, g_{n-k}$ and it can be obtained by measuring the observables $g_j$'s on $e\ket{\psi}$.
%One can then use the syndrome $\bs$ to find the most likely error vector $\e'$. If $\e'\in C+\e$,
%then the original state $|\psi \rangle$ can be reconstructed by applying $e'$ to the corrupted state.

% \rmark{ Also note that  errors in $\cS$ are not harmful.
% Thus the minimum distance $d$ of $\cC(\cS)$ is defined as the minimum weight of any element in
%${\cS^{\perp}}\setminus \overline{\cS},\label{eq:quantum duality}$
%where
%$
%\cS^{\perp}= \{h\in {\cG}_n: \langle h,g\rangle_{{\cG}_n}=0, \ \forall g\in \cS \}, %\label{eq:Sperp}
%$
%and $\overline{\cS}= \{cg: c\in\{\pm 1, \pm i\}, g\in\cS\}$.
%%which is the normalizer group  of $\cS$ in ${\cG}_n$.
%Then $\cC(\cS)$ is called an $[[n,k,d]]$ quantum stabilizer code.
%If there exists $g\in \cS$ with $\wt{g}< d$, $\cC(\cS)$ is called \emph{degenerate}; otherwise, it is \emph{nondegenerate}.}

%\cnote{transition words for a new section}

\section{Quantum Data-Syndrome Codes}\label{sec:quantum_data-syndrome}
Using a quantum stabilizer code with a corresponding check matrix $H$, quantum error correction can be done according to the error syndrome $\bs\in\mF_2^{n-k}$.
 %One of the most difficult problems in quantum error correction is, 
 However,  the syndrome $\bs$ itself
 could be measured with an error. So instead of the true vector $\bs\in\mF_2^{n-k}$, we may get, after measurement, a vector
 $\widehat{\bs}=\bs+\bz,$ for syndrome error $\bz\in \mF_2^{n-k}$.
 In other words, syndrome bits could be flipped.
 Herein we discuss stabilizer codes
that  are capable of correcting  both data errors and syndrome errors. %--\emph{quantum data-syndrome codes}.
To shorten notation, we will use $m\triangleq n-k$ in the following. 

%\subsection{Definition}
 The central idea is that the error syndrome of a measurement error on the $i$th syndrome bit is the vector
 $(0\cdots0 \underbrace{1}_i 0\cdots 0)$.
 Thus we define a new parity-check matrix
 \begin{equation}\label{eq:HDS}
 \hat{H}=[H\ I_{m}],
 \end{equation}
 where $I_{m}$ is considered as a matrix over $\mF_2$.
Define codes
 \begin{align*}
 D&=\{\bw=\bu \hat{H}: \bu\in \mF_2^m\}\subset \mF_4^n\times\mF_2^m, \mbox{ and }
 D^\perp=\{\bv: \bw \star \bv=0, \forall \bw\in D\}\subset \mF_4^n\times\mF_2^m,
 \end{align*}
where  $D^\perp$ is the dual code of $D$ with respect to the inner product:
%\vspace{-0.2cm}
%\begin{equation*}%\label{eq:inner_prod}
 $
 \x\star\by=\Tr^{\mF_4}_{\mF_2}\left(\sum_{i=0}^n x_i\overline{y}_i\right)+\sum_{j=0}^m x_{n+j}y_{n+j},
$
  for $ \x,\by\in \mF_4^n\times \mF_2^m.$ Therefore, a quantum stabilizer code is inherently capable of handling both data and syndrome errors.
  Fujiwara, \cite{Fuji14}, noticed that the error-correcting capabilities of $D^\perp$ depend on the choice of generators in~(\ref{eq:HDS}).
%In consequence of a bad choice of generators, even a single syndrome error may be uncorrectable. On the contrary, 
Choosing generators properly, we can get a code capable of correcting simultaneously  multiple
data and syndrome errors.

In addition, the error-correcting capabilities of $D^\perp$ can be further inhanced.
The standard approach to reduce the probability of syndrome measurement error is repeated syndrome measurement. That is, we repeat the syndrome measurement several times and take a majority vote. This is the same idea as in classical repetition codes. We propose a generalization of this idea by measuring additional stabilizers according to more powerful linear classical codes.

Let $C$ be an $[m+r,m]$ linear binary code with a generator matrix in the systematic form
\begin{equation}\label{eq:G(C)}
G_C=\left[I_{m}\ A\right],
\end{equation}
where $A=[a_{i,j}]$ is an $m\times r$ binary matrix. We define a new set of $r$ stabilizers ${\bf f}_j$ by
\begin{equation}\label{eq:f}
\f_j=a_{1,j}\bg_1+\cdots +a_{m,j}\bg_m, \mbox{ for }j=1,\dots,r.
\end{equation}
These $\f_j$ belong to the stabilizer group $\cS$, and
 can be measured without disturbing the underlying quantum codewords.
For this reason, we  call $C$ {\em the syndrome measurement (SM) code}.
 Let
 $
H'^T=\left[\begin{array}{ccc}
\f_1^T&
\cdots&
\f_r^T\end{array}\right].
$ 
With additional $r$ stabilizers being measured, it is equivalent to considering the code defined by the following parity-check matrix
 $
 \left[\begin{array}{ccc}
 H & I_{m} & 0 \\
H' & 0 & I_r
\end{array}\right].
$
This parity-check matrix can be transformed into the form
\begin{equation}\label{eq:HImAIr}
H\ds=\left[\begin{array}{ccc}
 H & I_{m} & 0 \\
0 & A^T & I_r
\end{array}\right].
\end{equation}
%where $A$ is a binary matrix.
We will say that (\ref{eq:HImAIr}) defines a quantum {\em data-syndrome code} $Q\ds$. It is convenient to define codes
\begin{align*}
 C\ds&=\{\bw=\bu H\ds: \bu\in \mF_2^{m+r}\}\subset \mF_4^n\times\mF_2^{m+r}, \mbox{ and }\\
 C\dsp&=\{\bv: \bw \star \bv=0, \forall \bw\in C\ds\}\subset \mF_4^n\times\mF_2^{m+r},
 \end{align*}
where  $C\dsp$ is the dual code of $C\ds$ with respect to the inner product:
%\vspace{-0.2cm}
\begin{equation}\label{eq:inner_prod}
 \x\star\by=\Tr^{\mF_4}_{\mF_2}\left(\sum_{i=1}^n x_i\overline{y}_i\right)+\sum_{j=1}^{m+r} x_{n+j}y_{n+j}.
 \end{equation}
Slightly abusing terminology, we will call $C\ds^\perp$ also a
 {\em data-syndrome} code. We will say that $Q\ds$ (or $C\ds^\perp$) has {\em length} $n$, {\em dimension} $k$, and {\em size} $2^k$. Note that a DS code of length $n$ and dimension $k$ encodes $k$ (information) qubits into $n$ (code) qubits.

 It is easy to see that $|C\ds|=2^{n-k+r}$ and $|C\dsp|=2^{2n+n-k+r}/|C\ds|=2^{2n}$, so the size of $C\dsp$ does not depend on a choice of $k$ and $r$. For a given matrix (\ref{eq:HImAIr}), we always can find vectors $\bg_{n-k+1},\ldots,\bg_n$ and 
 $\bh_1, \ldots,\bh_{n}$ over $\mathbb{F}_4$ such that $\bg_i*\bg_j=0$, $\bg_i*\bh_j=0$ for $i\neq j$ and  $\bg_i*\bh_i=1$. These vectors allow us to write down a generator matrix of $C\dsp$ in the following explicit form
\begin{equation}\label{eq:gen_matr_Cdsp}
G_{C\dsp}=\left[ \begin{array}{lll}
G & {\bf 0} & {\bf 0} \\
H_1 & {\bf 0} & {\bf 0} \\
H_2 & I_{n-k} & A
\end{array}\right], \mbox{ where }
G=\left[\begin{array}{c}
\bg_1\\
\vdots\\
\bg_{n}
\end{array}\right],
H_1=\left[\begin{array}{c}
\bh_{n-k+1}\\
\vdots\\
\bh_{n}
\end{array}\right],
H_2=\left[\begin{array}{c}
\bh_{1}\\
\vdots\\
\bh_{n-k}
\end{array}\right],
\end{equation}
and ${\bf 0}$s are all zero matrices of appropriate sizes.

We will say that a code defined by (\ref{eq:HImAIr}) is an $[[n,k:r]]$ DS-code. If the minimum distance $d$ (defined in   Section~\ref{sec:swe}) is known we will say that it is an $[[n,k,d:r]$ code.

\section{Syndrome measurement codes}\label{sec:SMcodes}
In this section, we consider some examples of stabilizer codes for which it is easy to find a efficient SM code that beat the repetitive syndrome measurement approach.

Suppose we are using an $[[n,k]]$ stabilizer code defined by a stabilizer group $\mathcal{S}$ with
generators $\bg_1,\ldots,\bg_{n-k}$.  Let $\bs = (s_1,\ldots, s_{n-k})$ be the correct syndrome
measurement outcomes. Let $\hat{\bs} = (\hat{s}_1,\ldots, \hat{s}_{n-k})$ be the (imperfect) syndrome bits output by Shor's syndrome extraction. Assume that the probability that an  $X$ or $Z$ measurement error occurs with probability $p_m$. Then the probability of incorrect measurement outcome on $\bg_j$ is
\begin{align}
p_{err}(%\bg_j)%&=\Pr(\mbox{error measurement of }
\bg_j)=\Pr(\widehat{s}_j\not=s_j)%\nonumber\\
&=\sum_{i \mbox{ \small{is odd}}} {\small{\wt}(\bg_j)\choose i}
p_m^i (1-p_m)^{\small{\wt}(\bg_j)-i} .
\label{eq:sj_error}
\end{align}

Now suppose that  an $[m+r,m]$ SM code $C$  is used.
Denote by $s_j$ and $z_j$ the results of correct measurement of
$\bg_j$ and $\f_j$, respectively.
It is not difficult to see that
\begin{equation}\label{eq:x}
\x=(s_1,\ldots, s_m,z_1,\ldots,z_{r})
\end{equation}
is a valid codeword of $C$. After the measurement of $\bg_1,\ldots,\bg_m$ and $\f_1,\ldots,\f_{r}$, we obtain a vector
\begin{equation}\label{eq:xhat}
\widehat{\x}=(\widehat{s}_1,\ldots,\widehat{s}_m,\widehat{z}_1,\ldots,\widehat{z}_{r}).
\end{equation}
\noindent We can correct quantum and syndrome errors simultaneously by decoding vector
$(\underbrace{0,\ldots,0}_{n \mbox{\small times}},\widehat{\x})$ using a decoder of $C\dsp$. Alternatively we can first correct syndrome errors by decoding vector $\widehat{\x}$ using a decoder of the SM code $C$, and next correct quantum errors. The latter  approach is typically simpler, though its performance is always suboptimal. In this section we consider this type of decoding.

Applying a decoding algorithm of $C$ to $\widehat{\x}$, we obtain bits $\tilde{s}_1,\ldots,\tilde{s}_m$.
For a given $C$ and its decoding algorithm, we define the syndrome decoding error and average syndrome decoding error, respectively, as
\begin{equation}\label{Pse1}
P_{se}=\Pr((s_1,\ldots,s_m)\not=(\tilde{s}_1,\ldots,\tilde{s}_m)),
\end{equation}
%and
\begin{equation}\label{eq:Psynderr1}
P_{SBER}={1\over m} \sum_{j=1}^m \Pr(\tilde{s}_j\not=s_j).
\end{equation}

%\noindent{\bf Remark}: we can also use a nonsystematic generator
%matrix $G$ in (\ref{eq:G(C)}) and  all our results easily extend. Usually, however, $\wt(\bg_j)<\wt({\bf
%	f}_i)$ and so $p_{err}(\bg_j)<p_{err}({\bf f}_i)$. Thus, the
%systematic form of $G$ allows us to minimize, the error probability
%for the first $m$ measurements.

The $l$-fold repeated syndrome measurement can be considered as a
particular case of an encoded syndrome measurement. It corresponds
to the SM code with generator matrix
$
G=[\underbrace{I_m \cdots I_m }_{l \mbox{ times}}].
$
Choosing a good SM code is not equivalent to finding a good $[m+r,m]$ linear code in the usual sense.
This is because for a typical $[m+r,m]$ code with a large minimum distance, the matrix $A$ in (\ref{eq:G(C)})
will have ``heavy" columns. This may result in that $\wt(\f_j)>> \wt(\bg_l)$ and therefore $p_{err}({\bf f}_j)>>p_{err}(\bg_i)$,
which, in turn, will lead to large $P_{se}$ and $P_{SBER}$.

Below we present several families of high rate quantum codes with
the property that all their stabilizers $\bg\in {\cal S}$ have the
same or almost the same weight and therefore any good linear codes
can be used for robust syndrome measurement.

Let $S_a$ be a generator matrix of the $[2^a-1,a,2^{a-1}]$ simplex
code. Let ${\cal S}_a$ be the $[[2^a-1,2^a-1-2a,3]]$ CSS code
defined by the generators
%$[S_l 0]$ and $[0 S_l]$.
$\left[\begin{array}{cc} S_a& 0\\
0 & S_a\end{array}\right].$
Any liner combination of the first (second) $a$
generators is a vector of weight $2^{a-1}$. Thus we can use any
good $[a+r,a]$
linear code for syndrome measurement of the first (second) $a$
syndrome bits. For example,   consider  ${\cal S}_3$, which  is the $[[7,1,3]]$ Steane
code. Let us use as an SM code the $[15,3]$ code $C$
defined by:
$$
G_C=\left[\begin{array}{l}
1 0 0     0        0     1     1     1     1     1     1   1     0     0     0\\
0 1 0     0          1     0     0     1     1     1     1   0     1     1     0 \\
0 0 1     1          0     1     1     0     0     1     1   0     1     1     1
\end{array}
\right].
$$
The code $C$ requires $15$ measurements, which is the same as for
the 5-fold repeated measurements of 3 bits. The corresponding
probabilities $P_{se}$ are shown in Fig.\ref{fig:Steane Code}. One
can see that the code $C$ has significantly lower $P_{se}$.

\begin{figure}[htb]
	\[	\includegraphics[scale=0.26]{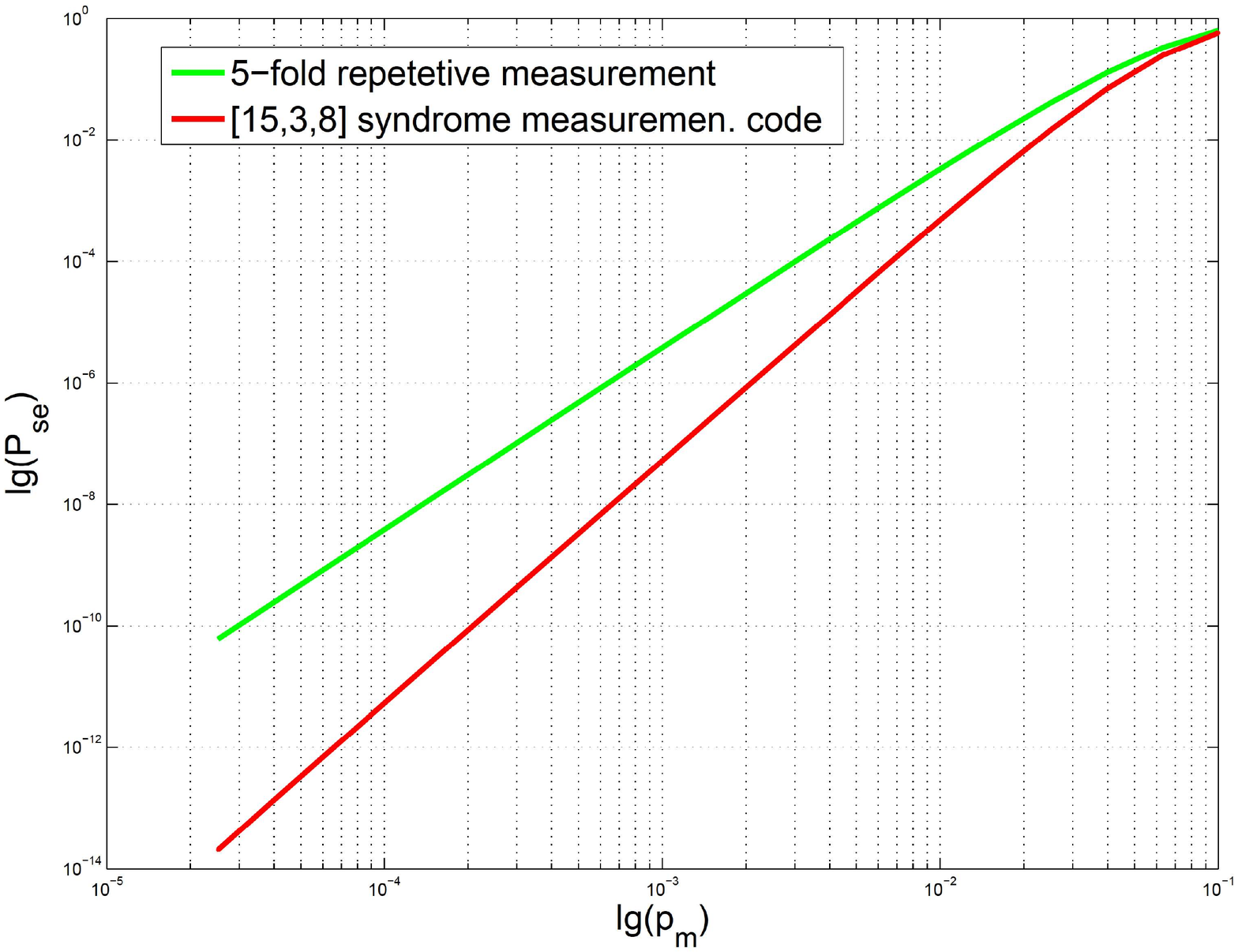}
	\]  
	\caption{The probabilities $P_{se}$ for the $[15,3]$
		code  $C$ and 5-fold Repeated Measurement of the syndrome of the
		 Steane code}  \label{fig:Steane Code}
\end{figure}

Another important family is the $[[n,n-2a,3]]$ quantum
Hamming codes ${\cal H}_a$ for $n=(4^a-1)/3$~\cite[V]{CRSS98}. It is not difficult to
prove that all generators of ${\cal H}_a$ have  weight $4^{a-1}$.

In \cite[Thm 11]{CRSS98} a family of $[[n,n-a-2,3]]$ codes with
$n=\sum_{i=1}^{(a-1)/2} 2^{2i+1}$ is defined for odd $a$. The generators of
these codes can have only weights $2^a-2$ and $2^a$.

\section{Mimimum Distance and Split Weight Enumerators}\label{sec:swe}

%The most important combinatorial parameter of an error correcting code is its minimum distance. So it is natural to define the minimum distance for DS codes.

Let $\e\ds=(\bg,{\bf 0})\in \mF_4^n\times \mF_2^{m+r},$ with $\bg \in C$.
 Since $\bg\in C$, we have $g\in \cS$ and thus $\e\ds$ is harmless.
 If $\e\ds=(\e,\bz)\in C\dsp\setminus \{(\bg,{\bf 0}): \bg\in C\}$, then $\e\not\in C$. Therefore the operator $e$ does not belong to $\cS$ and
 acts on $Q$ nontrivially. Since $H\ds\star \e\ds={\bf 0}^T$ by definition, we conclude that such $\e\ds$ is an undetectable and
 harmful error.
Naturally, the weight $\wt(\e,\bz)$ is defined as the number of its nonzero entries.
We define the {\em minimum distance} $d$ of $Q\ds$ (equivalently $C^\perp\ds$) as the minimum weight of any element in
 $$
C\dsp\setminus \{(\bg,{\bf 0}): \bg\in C\}.
 $$
It is not difficult to see that $Q\ds$ (or equivalently $C\dsp$) can correct any error  $\e\ds=(\e,\bz)$ (here we do not assume $\e\ds\in C\dsp$) with $\wt(\e)=t_{\tD},\wt(\bz)=t_{\tS}$ if
$
t_{\tD}+t_{\tS}<{d\over 2}.
$
Apparently the minimum distance of a DS code cannot be greater than that of the underlying stabilizer code.
%It is also clear that measuring more additional stabilizers can improve the accuracy of syndrome bits
%so that the minimum distance approaches its optimal, the minimum distance of the underlying stabilizer code.

%In contrast, if $w\tD+w\tS>{d\over 2}$, then there exist $\e\ds$ that are not correctable, i.e., there exists $\bv \in C\dsp$ such
%that $\dist({\bf 0},\e\ds)\ge \dist(\bv,\e\ds)$.
%
%Apparently, $d\ds$ is at most  the minimum distance of its underlying quantum stabilizer code
%$$d\ds\leq d.$$
%We say $Q\ds$ is an $[[n,k,r,d\ds]]\ds$ quantum DS code, where $r$ denotes the number stabilizers to be measured.
%The code $Q\ds$ (or equivalently $C\dsp$) can correct any error  $\e\ds=(\e,\bz)$ with $w\tD=\wt(\e),w\tS=\wt(\bz)$ and
%$$
%w\tD+w\tS<{d\ds\over 2}.
%$$
%If $w\tD+w\tS>{d\ds\over 2}$, then  $\e\ds$ is not correctable. That is, there exists $\bv \in C\dsp$ such
%that $\dist({\bf 0},\e\ds)\ge \dist(\bv,\e\ds)$.

% Let $\e\ds=(\bg,{\bf 0})\in C\ds,\bg\in\mF_4^n\setminus{\bf 0}$. From the definitions of codes $C\dsp$, $C\ds$, and   $C$,
% it follows that $\e\ds\in C\dsp$ and $\bg\in C$. Therefore the corresponding operator $g$ belongs to $\cS$ and thus $\e\ds$ is harmless.
  Define the split weight enumerators of $C\ds$ and $C\dsp$ by
 \begin{align}
B_{i,j}= B_{i,j}(C\ds)=|\{ & \bw\in C\ds: \wt(w_1,\ldots,w_n)=i,
 \wt(w_{n+1},\ldots,w_{n+m+r})=j\}|, \label{def:Bij}\\
B_{i,j}^\perp= B_{i,j}(C\dsp)=|\{ & \bw\in C\dsp: \wt(w_1,\ldots,w_n)=i, \wt(w_{n+1},\ldots,w_{n+m+r})=j\}|,  \label{def:Bijp}
 \end{align}
 respectively.
 The minimum distance $d$ of $Q\ds$ implies that
\begin{align}
B_{i,0}^{\perp}= \sum_{j=0}^{m+r} B_{i,j}, \text{ for } i=1, \dots, d-1. \label{eq:stabilizer_number}
\end{align}
We will say that $Q\ds$ is a {\em degenerate} quantum DS code if there exists $B_{i,j}>0$ for $i<d$.
Otherwise, it is a {\em nondegenerate} quantum DS code.
Clearly, we also have
 \begin{align}
 B^\perp_{i,0}\ge &\sum_{j=1}^{m+r} B_{i,j}, i=d,\ldots,n, \mbox{ and } \label{eq:Bperp>B} \\
 B_{0,0}=B_{0,0}^\perp=1 \mbox{ and } &B_{i,0}=0 \mbox{ for } i\ge 1.\label{eq:Bi0=0}
\end{align}
For $0\le i\le m+r$, let us define $d(i)$ as the smallest integer such that
\begin{equation}\label{eq:d_min_def_via_Bij1}
 B_{d(0),0}^\perp>\sum_{j=1}^{m+r} B_{d(0),j}, \mbox{ and } B_{d(i),i}^\perp>0, \mbox{ for } i=1,\ldots,m+r.
\end{equation}
Then the minimum distance of $Q\ds$ is
 \begin{equation}\label{eq:d_min_def_via_Bij2}
 d=\min_{0\le i\le m+r} d(i)+i.
 \end{equation}
Denote the $q$-ary Krawtchouk polynomial of degree $i$ by
\begin{equation}\label{eq:Kraw}
K_i(x;n,q)=\sum_{j=0}^i (-1)^j (q-1)^{i-j}{x\choose j}{n-x\choose i-j}.
\end{equation}
 We list  the properties of Krawtchouk polynomials needed in this work in Appendix~\ref{app:Kraw}. Their proof and other information on these polynomials can be found in~\cite{MS77,Lev95,AL99}.
Let $f(x,y)$ be a two-variable polynomial and its maximal degrees of $x$ and $y$ be $d_x\le n$ and $d_y\le m+r$, respectively. Then the following Krawtchouk expansion of this polynomial holds:
\begin{equation}\label{eq:KrawExpension}
f(x,y)=\sum_{i=0}^{d_x} \sum_{j=0}^{d_y} f_{i,j} K_i(x;n,q_1) K_j(y;m+r,q_2),
\end{equation}
where
\begin{equation}\label{eq:f_{i,j}}
f_{i,j}={1\over q_1^n q_2^{m+r}} \sum_{x=0}^n \sum_{y=0}^{m+r} f(x,y)K_x(i;n,q_1)K_y(j;m+r,q_2).
\end{equation}
Proofs of these equalities are straightforward generalizations of the proofs (see \cite[Chapter 5]{MRRW77}) for equivalent expressions for single variable polynomials.

%If $A_i$ is the weight distribution of an $[n,k]$ linear code over $\mF_q$ and $A_j^\perp$ is the weight distribution of its dual code, then, according to the MacWilliams identities, we have
%$$
%A_i={1\over \textcolor[rgb]{0.98,0.00,0.00}{q^{n-k}}} \sum_{j=0}^n A_j^\perp K_i(j; n,q).
%$$
In what follows we will need the following generalization of MacWilliams identities \cite{MS77}.
\begin{theorem}\label{thm:MaWil}
\begin{equation}\label{eq:MacW}
B_{x,y}={1\over {4^n}}\sum_{i=0}^n\sum_{j=0}^{m+r} B_{i,j}^\perp K_x(i;n,4) K_y(j;m+r,2).
\end{equation}
\end{theorem}
 A proof of this theorem can be found in Appendix~\ref{app:MaWil}.

Like \cite{CRSS98,LBW13,LA18}, for small $n$ one could apply linear programming techniques to obtain upper  bounds on the minimum distance of $[n,k:r]]$ DS codes. 
%or check feasibility of an $[[n,k,d:r]]$ code.
More explicitly,  we have the following linear  program: given $n,k,d,$ and $r$,
\begin{align*}
\mbox{Find \hspace{0.5cm}    }& \mbox{nonnegative intergers }B_{i,j}, B^{\perp}_{i,j} \notag\\
\mbox{subject to    \hspace{0.5cm} }& (\ref{eq:stabilizer_number}), (\ref{eq:Bperp>B}), (\ref{eq:Bi0=0}), \mbox{ and }(\ref{eq:MacW}).
\end{align*}
If there is no solution to this feasibility problem, it means that no $[[n,k,d:r]]$ DS code exists.
 %We demonstrate this in the following example.
\begin{example}
Let $n=7, m=6, r=0$,  and $d=3$. % (the $[[7,1,3]]$ Steane code).
Using MAPLE, we find out that the liner program has solutions.
This means that a $[[7,1,3]]$ code may be capable of fighting a syndrome bit error by measuring only six stabilizer generators.
Indeed, it is the case that a $[[7,1,3:6]]$ DS code exists as  shown by Fujiwara~\cite{Fuji14}.
\end{example}

\section{Upper Bounds on Unrestricted (Degenerate and Non-Degenerate) DS codes}\label{sec:bounds}
In this Section we propose a general method defined in Theorem \ref{thm:gen_upper_bound} for deriving upper bounds on the minimum distance of both non-degenerate and degenerate DS codes. Next we use this method for obtaining  several explicit bounds for DS codes with $r=0$. Theorem \ref{thm:gen_upper_bound} can be used for deriving bounds in the case of $r>0$, but this will be done in future work.

Let $1\le d\tD\le n$
be an integer and
$
\cN=\{(i,j): 0\le i\le n, 1\le j\le m+r\}.
$
Let also $\cA\subset \cN$ and $\overline{\cA}=\cN\setminus \cA$. We would like to upper bound {\em quantum code rate} $R=k/n$  under the conditions:
\begin{align}
B_{i,0}^\perp&=\sum_{j=0}^{m+r} B_{i,j}, i=0,\ldots,d\tD-1, \label{bound_cond1}\\
B_{i, j}^\perp&=0, (i,j)\in \cA.\label{bound_cond2}
\end{align}
\begin{theorem}
	Let $f(x,y)$ be an {\em arbitrary} polynomial with nonnegative coefficients $f_{i,j}$ that satisfies the conditions:
	\begin{align}
	f(x,0)\le 0, &\mbox{ if } x\ge d\tD, \mbox{ and} \label{eq:polin_cond1}\\
	f(x,y)\le 0, &\mbox{ if } (x,y) \in \overline{\cA}\label{eq:polin_cond2}.
	\end{align}
	\label{thm:gen_upper_bound}
	\begin{enumerate}
		
		\item For non-degenerate $C^\perp\ds$, it must hold that
		\be\label{eq:non-degBound}
		f(0,0)/ f_{0,0}\ge 2^{2n}.
		\ee
		\item For unrestricted $C^\perp\ds$, it must hold that
		\begin{equation}\label{eq:general_bound}
		\hspace{-0.5cm} \max\left\{ {f(0,0)\over f_{0,0}}, \max_{1\le x\le d\tD-1} {f(x,0)\over \min_{1\leq j\leq {m+r}} f_{x,j}}\right\} \ge 2^{2n}.
		\end{equation}
	\end{enumerate}
\end{theorem}
\proof
We prove the second claim.
Let $M=|C\dsp|=2^{2n}$. Using (\ref{eq:MacW}), (\ref{eq:KrawExpension}), and (\ref{eq:Bperp>B}),  we get
\begin{align}
&M\sum_{i=0}^{d\tD-1} \sum_{j=0}^{m+r} f_{i,j}B_{i,j}\le M \sum_{i=0}^n\sum_{j=0}^{m+r} f_{i,j}B_{i,j}\label{eq:gen_bound_line1}\\
=&M \sum_{i=0}^n\sum_{j=0}^{m+r} f_{i,j}{1\over M} \sum_{x=0}^n\sum_{y=0}^{m+r} B_{x,y}^\perp K_i(x;n,4)K_j(y;{m+r},2)\nonumber\\
%=&\sum_{l=0}^n\sum_{r=0}^m B_{l,r}^\perp f(l,r)\\
=&\sum_{x=0}^{n} B_{x,0}^\perp f(x,0)+\sum_{(i,j)\in \cA} B_{i,j}^\perp f(i,j)+\sum_{(i,j)\in \overline{\cA}}
B_{i,j}^\perp f(i,j)\nonumber\\
\le & \sum_{x=0}^{d\tD-1} B_{x,0}^\perp f(x,0)=\sum_{x=0}^{d\tD-1} \sum_{j=0}^{m+r} B_{x,j} f(x,0) \label{eq:gen_bound_last line}
%\sum_{l=0}^{2e\tD}\sum_{r=1}^{2e\tS} B_{l,r}^\perp f(l,r)=
%=\sum_{l=0}^{2e\tD} B_{l,0} f(l,0).
\end{align}
From this and (\ref{eq:Bi0=0}), we get
\begin{align}
2^{2n}\le &{\sum_{x=0}^{d\tD-1}\sum_{j=0}^{m+r} B_{x,j} f(x,0) \over \sum_{i=0}^{d\tD-1}\sum_{j=0}^{m+r} f_{i,j}B_{i,j}} 
\le  \max\left\{ {f(0,0)\over f_{0,0}}, \max_{1\le x\le d\tD-1} {f(x,0)\over \min_{1\leq j\leq {m+r}} f_{x,j}}\right\}.  \label{eq:max f(j)/fj}
\end{align}
\hfill\QED

%We use this theorem to get explicit upper bounds on $R=k/n$.
%\end{proof}

For getting a bound on the size of DS codes with minimum distance $d$, it suffices to choose
\be\label{eq:cA for codes with min dist d}
d\tD=d, \ \cA=\{(i,j): j\geq 1,  0\leq i+j\le d-1\},
\ee
and a polynomial $f(x,y)$ that satisfies constraints (\ref{eq:polin_cond1}) and (\ref{eq:polin_cond2}).
In the following subsections, we discuss two polynomials and their corresponding bounds on DS codes with $r=0$.

On the other hand, we also have upper bounds for a DS code inherited from its underlying stabilizer code.
Let $Q\ds$
be a DS code defined by (\ref{eq:HImAIr}) ($r=0$) with minimum distance $d(Q\ds)$.
Let $C$ be the $[n,n-k]$ code with generator matrix $H$ used in (\ref{eq:HDS}) and $C^\perp$ be its dual code.
Let $Q$ be the $[[n,k]]$ stabilizer code defined by $C$, and let its minimum distance be $d(Q)$.
From (\ref{eq:gen_matr_Cdsp}) it is easy to see that vectors of the form
$
(\bv,{\bf 0}), \mbox{ where }\bv\in C^\perp,~{\bf 0}=(\underbrace{0,\ldots,0}_{n-k}),
$
form a subcode of $C\ds^\perp$ and therefore $d(Q\ds)\le d(Q)$. Thus any upper bound on degenerate $[[n,k]]$ stabilizer code $Q$ is also an upper bound on the minimum distance of degenerate $[[n,k;0]]$ DS code. The same is true for non-degenerate codes.

\subsection{Singleton Bound}\label{sec:SinB}

As we mentioned in Section \ref{sec:quantum_data-syndrome} code $C\dsp$ has size $2^{2n}$, and if $\bv=(v_1,\ldots,v_n,w_1,\ldots,w_{n-k})\in C\dsp$ then  $v_i\in \mathbb{F}_4$ and $w_i\in \mathbb{F}_2$.
This leads to the Singleton bound for nondegenerate DS codes.
\begin{theorem}\label{thm:SingletonBounNonDeg} In any nondegenerate $[[n,k,d:0]]$ DS code we have
\begin{equation}\label{eq:SingleNonDeg}
k\leq n-2(d-1).
\end{equation}
\end{theorem}
The proof of this theorem is a simple generalization of the well known case of codes over\,$\mathbb{F}_q$~\cite{Sin64}. 
%For self-completeness of the presentation, we put it into Appendix~\ref{app:singletonNonDeg}.

In \cite{AL99} several upper bounds for degenerate stabilizer codes have been derived. In particular, the Singleton bound $k\le n-2(d(Q)-1)$ has been proven. Thus we conclude that bound (\ref{eq:SingleNonDeg}) also holds for degenerate DS codes.

It is instructive to prove this result using (\ref{eq:general_bound}). To do this, we first note that
if $f(x,y)=0$ for $y\ge 1$, then the coefficients $f_{i,j}$ do not depend on $j$. Indeed, let $f(x,y)=g(x)\delta_{y,0}$ and
$
f(x,0)=g(x)=\sum_{i=0}^n g_i K_i(x; n,4).
$
Then, according to (\ref{eq:KrawExpension}) and (\ref{eq:K0=1}), we have
\begin{equation}\label{eq:f_i,j_const_in_j}
f_{i,j}={1\over 4^n 2^m} \sum_{x=0}^n g(x) K_x(i; n,4)\sum_{y=0}^m \delta_{y,0}K_y(j;m,2)={1\over 2^m} g_i K_0(j;m,2)={1\over 2^m} g_i.
\end{equation}

\begin{theorem}\label{thm:SingletonBoun} ({\bf Singleton Bound}) For an unrestricted (non-degenerate or degenerate)  DS code, we have
\begin{equation}\label{eq:SingleB_degen}
k\leq n-2(d-1).
\end{equation}
\end{theorem}
\proof
We will use the polynomial
\begin{equation}\label{eq:single_polynom}
f(x,y)=\frac{4^{n-d+1}2^m}{{n\choose d-1}} {n-x \choose n-d+1}\delta_{y,0}.
\end{equation}

Using (\ref{eq:f_{i,j}}) and (\ref{eq:(n-i,n-j)K_i(x)}), we obtain
\begin{equation}\label{eq:singleton_f_ij}
f_{i,j}={1\over 4^n2^m} {4^{n-d+1}2^m\over {n\choose d-1}} \sum_{x=0}^n {n-x\choose n-d+1} K_x(i;n,4) \sum_{y=0}^m
\delta_{y,0} K_y(j;m,2) ={ {n-i\choose d-1}\over {n\choose d-1}}\ge 0, \forall i,j.
\end{equation}
It is easy to see that $f_{i,j}\ge 0$ and $f(i,0)=0$ for $i\ge d$. Simple computations show that
 $f(0,0)/f_{0,0}>f(l,0)/f_{l,j}=f(l,0)/f_{l,0}$ for $1\le l\le d-1$. Finally, 
 $f(0,0)/f_{0,0}=4^{n-d+1}2^m$.  \QED

This approach gives us additional information on DS codes achieving the Singleton bound.  If a DS code, say $Q_{\mathrm{DS,MDS}}$,  meets the Singleton bound, then in (\ref{eq:gen_bound_line1}) we must have equality. Noticing that $f_{i,j}=0$ (defined in (\ref{eq:singleton_f_ij})) for $i> n-d+1$, we conclude that $Q_{\mathrm{DS,MDS}}$
must have $B_{i,j}=0$ for $d\le  n-d+1$ and $j\ge 0$. In (\ref{eq:gen_bound_last line}) we always have equality since $f(x,0)=0$ if $x\ge d$. Finally, in order to have
$$
{\sum_{x=0}^{d\tD-1}\sum_{j=0}^m B_{x,j} f(x,0) \over \sum_{i=0}^{d\tD-1}\sum_{j=0}^m f_{i,j}B_{i,j}}
={\sum_{x=0}^{d\tD-1}f(x,0) \sum_{j=0}^m B_{x,j}  \over \sum_{i=0}^{d\tD-1}f_{i,0} \sum_{j=0}^m B_{i,j}}
=f(0,0)/f_{0,0}
$$
in (\ref{eq:max f(j)/fj}),
code $Q_{\mathrm{DS,MDS}}$ must have $\sum_{j=0}^m B_{x,j}=0$ for $1\le x \le d\tD-1$. Thus $B_{x,j}=0$ for $1\le x\le n-d\tD-1$ and $j\ge 0$.
This means that any generator $\bg$ of $Q_{\mathrm{DS,MDS}}$ should have large weight, $\wt(\bg)\ge n-d$. Hence such $Q_{\mathrm{DS,MDS}}$ will have large syndrome measurement error.

Extensive research have been conducted on construction of quantum codes meeting the Singleton bound (see for example \cite{MDS1}, \cite{MDS2}, \cite{MDS3}, \cite{MDS4}, references within, and numerous other papers on this subject). The above result however shows that such codes most likely will not be useful for practical applications due to their large syndrome measurement error probability.

%We can use $f_{i,j}$ to choose $u_{i,j}$ and $g_j$ as well.
%I checked some examples of parameters and it works.

%\textbf{Remark:} This bound is better when $m=n-k$, since it does not reflect any additional information when increasing $m$.
%But we can tell the effect of increasing $m$ by the (dual) LP bound.
\subsection{Hamming Bound}
Let $C^\perp\ds$ be a non-degenerate DS code with minimum distance $d=2t+1$.
Standard combinatorial arguments (see \cite{Fuji14}) lead to  that $k\le \tilde{k}$, where
$\tilde{k}$ is the largest integer such that
\begin{align}\label{eq:HBnondeg}
 2^{2n} & \leq {4^n2^{n-\tilde{k}}\over \sum_{i=0}^{t} {n\choose i}3^i \sum_{j=0}^{t-i} {n-\tilde{k}\choose j}} .
\end{align}
This is the {\em Hamming Bound  for non-degenerate} DS codes. Below we show that this bounds also holds for degenerate DS codes if $n$ is sufficiently large.

Let $d\tD$ and $\cA$ be defined as in (\ref{eq:cA for codes with min dist d}).
\begin{lemma}\label{lem:Ham_bound}
	For a positive integer $\lambda$,  let $f^{(k)}(x,y)$ be the polynomial defined by the coefficients
	$$
	f^{(k)}_{i,j}=\left(\sum_{a=0}^{t}\sum_{g=0}^{t-\lambda a}  K_a(j;m,2) K_g(i;n,4) \right)^2.
	$$
	Then
	$$\displaystyle 
	f^{(k)}(x,y)=4^n2^m \sum_{a=0}^{t}\sum_{b=0}^{t} \beta(y,a,b) \sum_{g=0}^{t-\lambda a} \sum_{h=0}^{t-\lambda b} \sum_{w=0}^{n-x} \alpha(x,g,h,w),
	$$
	where
	\begin{align}
	\alpha(x,g,h,w)=&{x\choose 2x+2w-g-h}{n-x\choose w}{2x+2w-g-h \choose x+w-h} 2^{g+h-2w-q}3^w, \label{eq:alpha}
	\mbox{ and }\\
	\beta(y,a,b)=&{m-y\choose (a+b-y)/2}{y\choose (a-b+y)/2}. \label{eq:beta}
	\end{align}
\end{lemma}
A proof can be found in Appendix~\ref{app:Hamming}.

Note that in the above lemma,  $\lambda$ is a parameter over which we will optimize our bound. Next we give an important property of the polynomial $f^{(k)}(x,y)$.
\begin{lemma}\label{lem:f(x,y)=0}
For $x+y\ge d$, we have $f^{(k)}(x,y)=0$.
\end{lemma}
 The proof is given in Appendix~\ref{app:f(x,y)=0}.

Thus $f^{(k)}(x,y)$ satisfies constraints (\ref{eq:polin_cond1}) and (\ref{eq:polin_cond2}) and hence we can use it for obtaining a bound on the minimum distance of DS codes.
Choosing $\lambda=1$, we get a polynomial with $f^{(\tilde{k})}(0,0)/f^{(\tilde{k})}_{0,0}$ equal to the right hand side
of (\ref{eq:HBnondeg}). Numerical computations show that for large $n$, the first entry in the set defined in (\ref{eq:general_bound}) dominates. Thus, for
large $n$, this polynomial gives the Hamming bound (\ref{eq:HBnondeg}) for unrestricted (non-degenerate and degenerate) DS codes.  The ``disatvantage"
of this polynomial is that its coefficients $f^{(k)}_{i,j}$ may agresively decrease with $j$, which for certain parameters makes $\min_{1\le j\le m} f(x,j)$ in (\ref{eq:general_bound}) being very small, that results in a loose bound.

If we choose $\lambda=t+1$,   we get $f^{(k)}(x,y)=f(x)\delta_{0,y}$, where $f(x)$ is the polynomial with  
$f_i=(\sum_{g=0}^t K_g(i;n,4) )^2$, that is the polynomial that leads to the Hamming bound for classical codes
over $\mF_4$, see \cite[Chapter 17]{MRRW77}. For this polynomial, the value $f^{({k})}(0,0)/f^{({k})}_{0,0}$ is larger than in the case of $\lambda=1$. However, its advantage is that its coefficients $f^{(k)}_{i,j}$ do not decrease with $j$ (in fact they do not depend on $j$), which often leads to  better bound than with $\lambda=1$.

%Minimizing over $\lambda$ we get the following result.
\begin{theorem}({\bf Hamming bound for unrestricted DS codes}.) For an unrestricted DS code, we have
$k\le \bar{k}$, where
 $\bar{k}$ is the largest integer such that
\begin{equation}\label{eq:HBdeg}
\min_{1\le \lambda \le t+1} \max\left\{ {f^{(\bar{k})}(0,0)\over f^{(\bar{k})}_{0,0}}, \max_{1\le x\le d\tD-1} {f^{(\bar{k})}(x,0)\over \min_{1\leq j\leq m} f^{(\bar{k})}_{x,j}}\right\} \ge 2^{2n}.
\end{equation}
\end{theorem}
For $d=7$, the Hamming bounds (\ref{eq:HBnondeg}) and (\ref{eq:HBdeg}) are shown in Fig.~\ref{fig:Hamming bound}. For small values of $n$,
bound (\ref{eq:HBdeg}) is only marginally weaker than (\ref{eq:HBnondeg}), and for $n\ge 36$, these bounds coincide.
We observed the same behavior for other values of $d$. So we make the following conjecture.
\begin{conjecture} For any $d$, there exists $n(d)$ such that for $n\ge n(d)$, the Hamming bound (\ref{eq:HBnondeg}) holds for unrestricted
DS codes.
\end{conjecture}

\begin{figure}[h]
%\vspace{-3.5cm}
\[
\includegraphics[width=7.5cm]{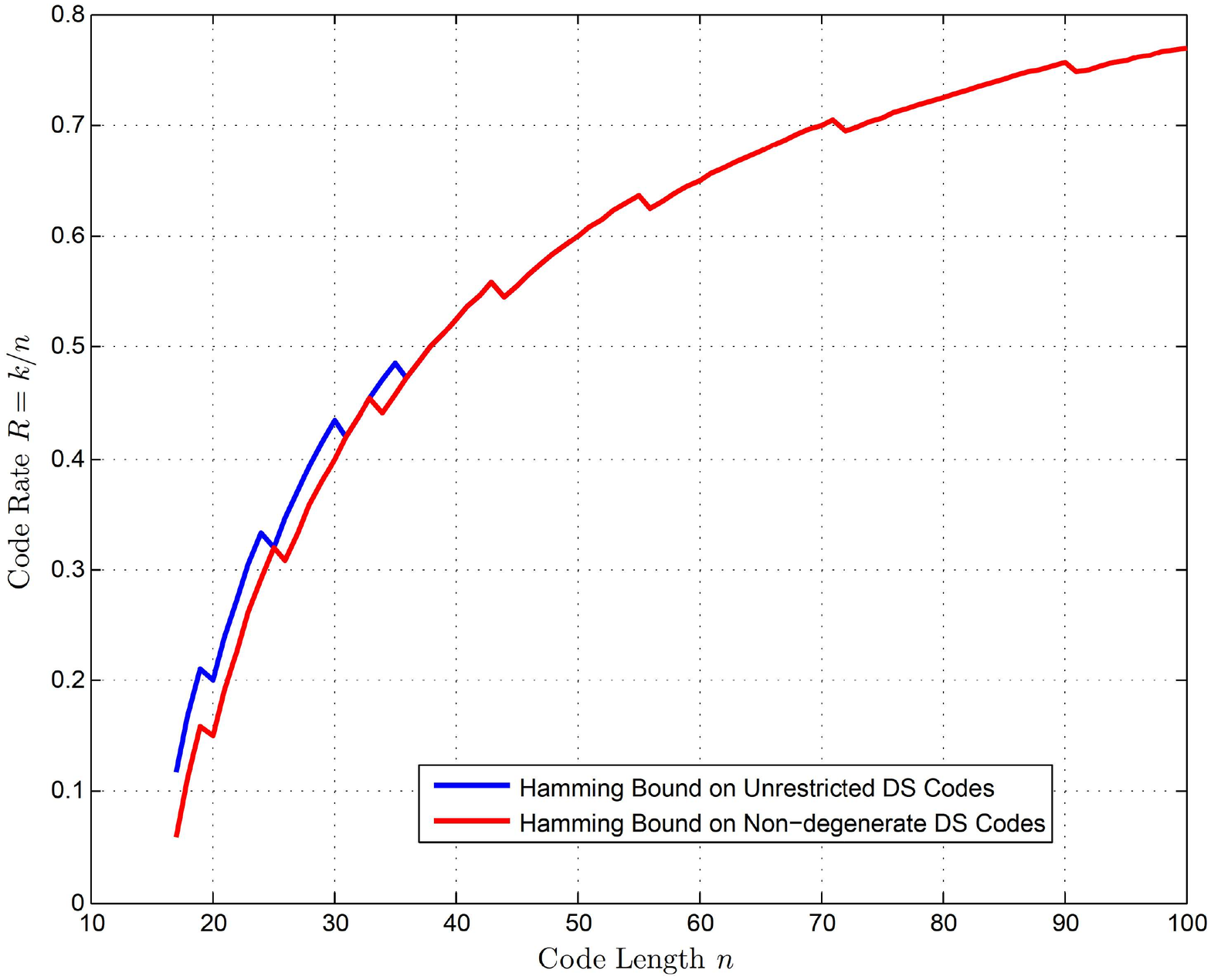}
\]
%\vspace{-4.0cm}
\caption{Hamming Bounds for nondegenerate and unrestricted DS codes, $d=7$.
 } \label{fig:Hamming bound} %\label{fig:distillation_repetition code}
%\vspace{-0.2cm}
\end{figure}

In \cite{Fuji14} Fujiwara obtained a {\em hybrid Hamming bound} for nondegenerate DS codes that can correct any
 $t\tD$ data and $t\tS$ syndrome errors:  $k\le \hat{k}$, where $\hat{k}$ is the largest integer  such that
\begin{equation}\label{eq:HybridHamBound}
 2^{2n}\le {2^{2n}2^{n-\hat{k}}\over \sum_{i=0}^{t\tD}\sum_{j=0}^{t\tS} {n \choose i}3^i {n-\hat{k}\choose j}}.
 \end{equation}
We can also derive this hybrid bound using Theorem \ref{thm:gen_upper_bound} with $\cA=\{ (i,j): 0\le i\le 2t\tD \mbox{ and } 1\le j\le 2t\tS\}$, and polynomial
$$
	f^{(k)}(x,y)=4^n2^m \sum_{i=0}^{t\tD}\sum_{j=0}^{t\tD} \sum_{h=0}^{n-x} \alpha(x,i,j,h) \times \sum_{u=0}^{t\tS}\sum_{v=0}^{t\tS}
	\beta(y,u,v).
	$$
Tedious but straightforward computations show that $f^{(k)}_{i,j}\ge 0$, $f^{(k)}(x,y)=0$ if $(x,y)\in \overline{\cA}$,
and that $f^{(\hat{k})}(0,0)/f^{(\hat{k})}_{0,0}$ is equal to the right hand side of~(\ref{eq:HybridHamBound}).

Thus we obtained a different proof of (\ref{eq:HybridHamBound}). We cannot use this polynomial
for degenerate  DS codes, since for some $i\le d\tD-1$, we have  $f^{(k)}_{i,j}=0$. Finding good polynomials for
deriving hybrid bounds on degenerate  DS codes is an open problem.

\subsection{Asymptotic Bounds}
In this subsection, we consider the asymptotic regime in which both the code length $n$ and the number of information qubits $k$ tend to infinity but the code rate $R=k/n$ remains  constant.

It is instructive to consider the Hamming bound (\ref{eq:HBnondeg}) for nondegenerate DS codes in this regime. In order of doing this we have to find the leading term of the denominator of~(\ref{eq:HBnondeg}).

Recall that if $v$ grows linearly with $n$ and $a_{j^*}>a_j,j\not =j^*$, then
\begin{equation}\label{eq:exponent_largest}
{1\over n} \log 2 \sum_{j}^v 2^{n a_j} = a_{j^*}+o(1).
\end{equation}
Let $H(x)=-x\log_2(x)-(1-x)\log_2(1-x)$ be the binary entropy function.
%(In the following the logarithm will be base 2 with the base omitted.)
Denoting $\xi=j/n, \tau=t/n$, and $\iota=i/n$ and using (\ref{eq:exponent_largest}), we get for the second sum of the denominator of\,(\ref{eq:HBnondeg})
$$
{1\over n} \log_2 \sum_{j=0}^{t-i} {n-\tilde{k}\choose j}= \frac{1}{n}\max_{0\le \xi \le \tau-\iota} \log_2 2^{(n-\tilde{k})H(\xi n /(n-\tilde{k})}+o(1)=\max_{0\le \xi \le \tau-\iota} (1-R)H\left({\xi\over 1-R}\right)+o(1),
$$
where $o(1)$ is a function that tends to $0$ as $n$ increases
and we have used Stirling's approximation   that ${1\over n} \log_2 {n \choose i} = H(i/n)+o(1)$ \cite{MS77}.
This function achieves its maximum at $\xi= \frac{1}{2} (1-R)$. However, according to the Singleton bound, the relative distance $\delta\triangleq \frac{d}{n}\le \frac{1}{2} (1-R)$ and therefore
$$
\tau=t/n=\delta/2 \le \frac{1}{4}(1-R) \le \frac{1}{2}(1-R).
$$
Thus the maximum is achieved at $\xi= \tau-\iota$. Hence for the denominator of (\ref{eq:HBnondeg}) we have
\begin{align}
{1\over n} \log_2 \sum_{i=0}^{t} {n\choose i}3^i \sum_{j=0}^{t-i} {n-\tilde{k}\choose j} =
\max_{0\le \iota \le \tau} H(\iota)+\iota \log_2 (3) +(1-R)H((\tau-\iota)/(1-R))+o(1).
\end{align}
Taking the derivative and finding its roots, we conclude that the maximum is achieved at
\begin{equation}\label{eq:iota*}
\iota^*=1-{1\over 4}R+{1\over 2}\tau - {1\over 4} \sqrt{16-8R-8\tau+R^2-4R\tau+4\tau^2}.
\end{equation}
It is not difficult to show that $\iota^*$ is always smaller than $\tau$. Thus the exponent of the denominator of (\ref{eq:HBnondeg}) is
$$
H(\iota^*)+\iota^* \log_2 (3) +(1-R)H((\tau-\iota^*)/(1-R))+o(1).
$$
The exponents of the left part and the numerator of (\ref{eq:HBnondeg}) are $2$ and
$
{1\over n} \log_2 4^n 2^{n-\tilde{k}} =3-R,
$
respectively. Combining the above results, we obtain the following theorem.
\begin{theorem}\label{eq:HamBound_asympt_nondegen}
For a given relative distance $\delta$, the code rate $R$ cannot exceed the root, say $R_{Ham,nondeg}(\delta)$, of
\begin{equation}\label{eq:ham_bound_nondegen_asympt}
H(\iota^*) +\iota^* \log_2 (3) +(1-R)H((\delta/2-\iota^*)/(1-R))+R-1= 0.
\end{equation}
\end{theorem}

 In \cite{AL99} the Hamming and so-called first linear programming (LP1) bounds have been derived in asymptotic form for unrestricted (degenerate and non-degenerate) quantum codes: 
 %They have the following forms
\begin{align}
R\le 1-\delta/2 \log_2 (3)-H(\delta/2)+o(1), \mbox{ for } 0\le \delta\le 1/3, \mbox{ (Hamming)} \label{eq:HammingBoundDegener}\\
R\le H(w)+w\log_2 (3)-1+o(1),~w={3\over 4} -{1\over 2}\delta -{1\over 2}\sqrt{3\delta(1-\delta)}, \mbox{ for } 0\le \delta\le 0.3152.
\mbox{ (LP1) } \label{eq:LP1BoundDegener}
\end{align}
The Hamming bound was obtained by applying the polynomial $f_{Ham}(x)$ defined by its coefficients
$f_{i}=K_{\delta/2}(i)^2$. LP1 bound was obtained with the help of the polynomial
$$
f_{LP1}(x)= {1\over a-x} (K_{t+1}(x)K_t(a)-K_t(x)K_{t+1}(a))^2,
$$
where $t={\delta\over 2}n$ and $a$ is a real number located between the first roots $r_{t+1}$ and $r_t$ of $K_{t+1}(x)$ and $K_t(x)$, and chosen so that $K_t(a)/K_{t+1}(a)=-1$.

As we mentioned prior to Section \ref{sec:SinB}, any bound on degenerate quantum codes is also a bound on degenerate DS codes with the corresponding $n$ and $k$. Thus bounds (\ref{eq:HammingBoundDegener})  and (\ref{eq:LP1BoundDegener}) also hold for unrestricted (degenerate or non-degenerate) DS codes. Note that these bounds can be also obtained using Theorem \ref{thm:gen_upper_bound} and polynomials
$
f(x,y)=f_{Ham}(x)\delta_{y,0} \mbox{ and } f(x,y)=f_{LP1}(x)\delta_{y,0}.
$
As we showed in (\ref{eq:f_i,j_const_in_j}) the coefficients $f_{i,j}$ of these polynomials do not depend on $j$.

The bounds (\ref{eq:ham_bound_nondegen_asympt}), (\ref{eq:SingleB_degen}), (\ref{eq:HammingBoundDegener}), and (\ref{eq:LP1BoundDegener}) for unrestricted DS codes are shown in Fig. \ref{fig:aympt bound}. One can see that at certain interval the Hamming bound for non-degenerate quantum codes beats all the bounds for degenerate DS codes.

\begin{figure}[ht]
%\vspace{-2.5cm}
\[
\includegraphics[scale=0.28]{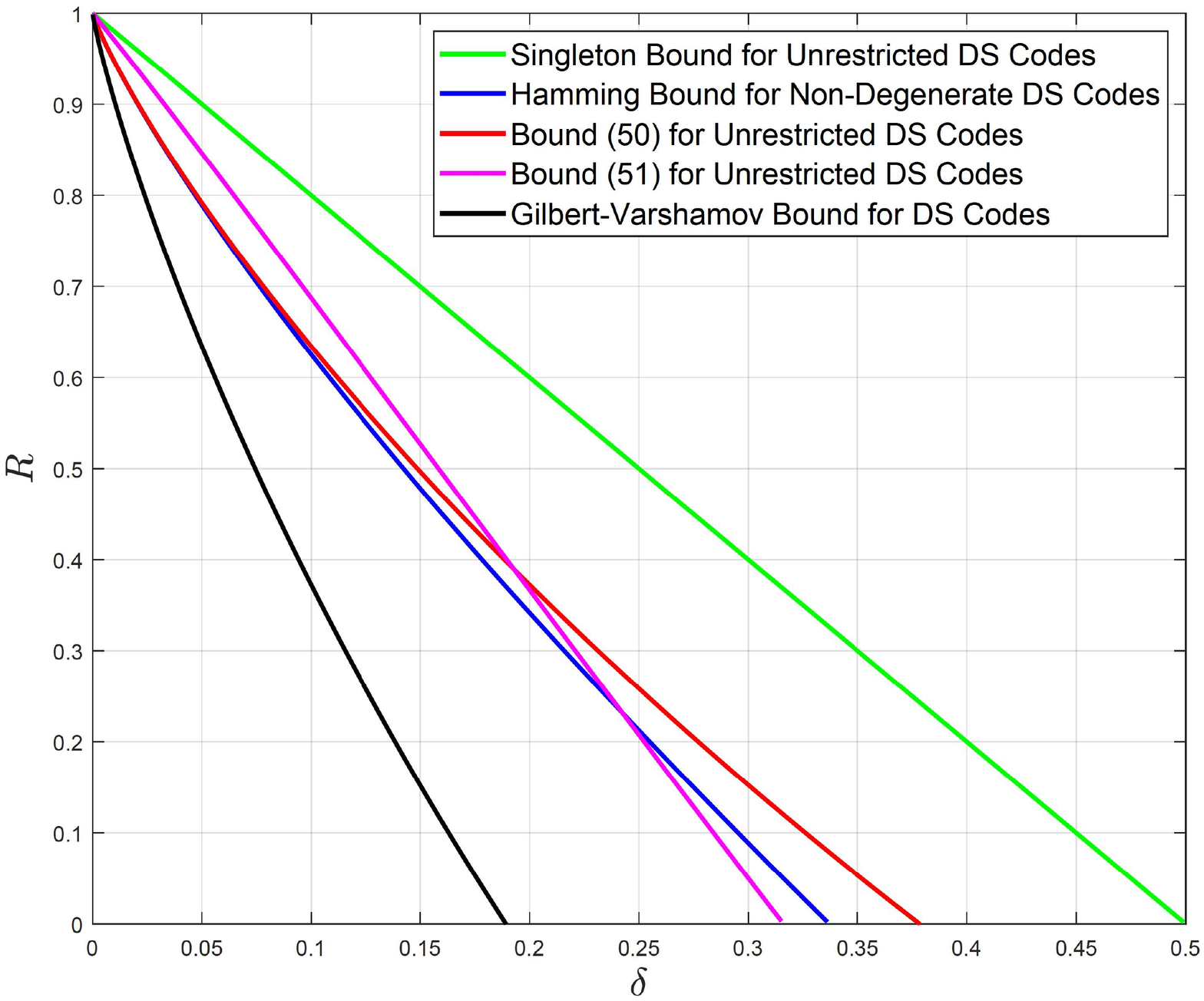}
\]
%\vspace{-3.3cm}
\caption{Upper Bounds on Unrestricted DS codes and an Achievability Bound on DS Codes}
\label{fig:aympt bound} %\label{fig:distillation_repetition code}
%\vspace{-0.2cm}
\end{figure}
\cnote{at the end, we  may have to change the caption in the figure as equations are numbered differently}
\anote{I will redo the picture soon}

 It looks natural to try to improve bounds (\ref{eq:HammingBoundDegener}), and (\ref{eq:LP1BoundDegener}) by using polynomials $f(x,y)$ whose coefficients $f_{i,j}$ depend on both indices $i$ and $j$. At this moment we did not find such polynomials and leave this as an interesting open problem.

\section{Random DS Codes} \label{sec:random}
The enumerators $B_{i,j}^\perp$ define the decoding error probability of a DS code in a number of communication/computational scenarios, similar to \cite{A14}. Below we study the behavior of   $B_{i,j}$ and $B_{i,j}^\perp$ of random DS codes.
In particular, we are interested in how the normalized minimum distance   $d(r)/n$ depends on the ratio $r/n$   when $n\rightarrow\infty$.

We will consider the ensemble $\cE_{n,k,r}$ of  $C\ds$  codes defined by matrices of the form (\ref{eq:HImAIr}) with $r\le n-k$ and full rank matrices $A$, i.e., $\mbox{rank}(A)=r$.  We will use this ensemble to show that the minimum distance of random DS codes with a relatively small $r\le n-k$ achieves the Gilbert-Varshamov bound of stabilizer codes \cite{ABKL00b}.

%\cnote{I changed the notation of ensemble to $\cE$ because $\cS$ is used for stabilizer groups.}
  Let $\cE_{n,k,r}^\perp$ be the ensemble of $C\dsp$ codes that are dual to codes from $\cE_{n,k,r}$. Note that $|\cE_{n,k,r}^\perp|=|\cE_{n,k,r}|$.
Define the average enumerators (weight distribution)  of codes   from $\cE_{n,k,r}$ and $\cE_{n,k,r}^\perp$, respectively, by
\begin{align*}
\overline{B}_{i,j}&={1\over |\cE_{n,k,r}|} \sum_{C\in \cE_{n,k,r}} B_{i,j}(C),\mbox{ and }\overline{B}_{i,j}^\perp={1\over |\cE_{n,k}^\perp|} \sum_{C^\perp\in \cE_{n,k}^\perp} B_{i,j}(C^\perp),
\end{align*}
where $B_{i,j}(C)$ and $B_{ij}(C^\perp)$ are defined in (\ref{def:Bij}) and (\ref{def:Bijp}).
The following theorem finds these weight distributions explicitly.
\begin{theorem}\label{thm:aver_weight_distribution}
For $1\le i\le n \mbox{ and  }1\le j \le m+r$, we have
\begin{align}
\overline{B}_{0,0}=&1,\ % \nonumber\\
\overline{B}_{i,0}=0,\notag \\ %\label{eq:Bi0} \\
\overline{B}_{0,j}=&{1\over 2^m-1}\left({m+r\choose j}- {m\choose j}-{r\choose j}\right),\label{eq:B0j}\\
\overline{B}_{i,j}=&{1\over (4^n-1)(2^m-1) }{n\choose i}3^i
\left( (2^m-2){r+m\choose j}+ {m\choose j}+{r\choose j}\right),\label{eq:Bij}\\
\overline{B}_{i,0}^\perp=&{1\over (4^n-1)(2^m-1)}{n\choose i}3^i
\left(4^n-2^m+1-{4^n\over 2^m}\right),\label{eq:Bd_i0}\\
\overline{B}_{i,j}^\perp=&{4^n\over (4^n-1)2^{r+m}(2^m-1)}{n\choose i}3^i
\left({m+r\choose j}2^m-{r\choose j}2^m-{m\choose j}2^r\right),\label{eq:Bd_ij}\\
\overline{B}_{0,0}^\perp=&1,\ \overline{B}_{0,j}^\perp=0.\label{eq:Bd_0j}
\end{align}
\end{theorem}
A combinatorial proof of this result can be found in Appendix~\ref{app:avg}.

Let us now consider the asymptotic case when the code length $n\rightarrow \infty$.
Again, denote
$
\iota=i/n,~\xi=j/n,\mbox{ and }\rho=r/n.
$
For a DS code with $B_{i,j}$ and $B_{i,j}^\perp$, we define
$
b_{\iota,\xi}={1\over n}\log_2  {B}_{\lfloor \iota n\rfloor,\lfloor \xi n\rfloor}, b_{\iota,\xi}^\perp={1\over n}\log_2  {B}_{\lfloor \iota n\rfloor,\lfloor \xi n\rfloor}^\perp, \mbox{ and } \delta=d/n,
$
where the minimum distance $d$ is defined by (\ref{eq:d_min_def_via_Bij1}) and (\ref{eq:d_min_def_via_Bij2}).
%\cnote{$B$ and $B^\perp$ are given in the proof?}
%\cnote{How about we move all the $b_{\iota,\xi}$ parts to Appendix and only keep  (\ref{eq:d_DS_GV}) in Theorem~\ref{thm:asympt_aver_weight_distribution}?}
%\anote{I would not do this. The behavior of weight enumerators can be interesting on its own, at least this the case in classical coding theory, and those numerators also allow obtaining upper bound on the probability of decoding error.}
Denote by $d_Q$ the minimum distance of a generic quantum  code.
In \cite{ABKL00b}, it was shown that there are quantum codes, and quantum stabilizer codes in particular, whose normalized minimum distance $\delta_Q=d_Q/n$  is at least as large as the quantum Gilbert-Varshamov (GV) bound $\delta_{GV}(R)$. This bound  is defined by the equation
$$
 H(\delta_{GV}(R))+\delta_{GV}(R)\log_2(3)=1-R.
$$
In the next theorem, we prove that there exist DS codes whose weight distributions $B_{i,j}$ and $B_{i,j}^\perp$ are upper bounded by the analytical expressions presented in the theorem for all $i$ and $j$, and present a GV bound $\delta_{DS,GV}$ for such codes.
\begin{theorem}\label{thm:asympt_aver_weight_distribution}
For $r\le m$, there exist DS codes with rate $R$ and $b_{\iota,\xi}\le \overline{b}_{\iota,\xi}$ and $b_{\iota,\xi}^\perp \le \overline{b}_{\iota,\xi}^\perp$ and
\begin{align}
\overline{b}_{0,\xi}&=(1-R+\rho)H\left({\xi\over 1-R+\rho}\right)-1+R\label{eq:b_0j}+ o(1),\\
\overline{b}_{\iota,\xi}&=H(\iota)+\iota\log_2(3)+(1-R+\rho)H\left({\xi\over 1-R+\rho}\right)-2+ o(1),\label{eq:b_ij}\\
\overline{b}_{\iota,\xi}^\perp&=H(\iota)+\iota\log_2(3)+(1-R+\rho)H\left({\xi\over 1-R+\rho}\right)-(1-R+\rho)+ o(1),\label{eq:b^perp_ij}\\
\overline{b}_{\iota,0}^\perp&=H(\iota)+\iota\log_2(3)-1+R+ o(1),\label{eq:b^perp_i0}
\end{align}
and the normolized minimum distance $\delta_{\mathrm{DS}}\ge \delta_{\mathrm{DS,GV}}(R,\rho)$, where
\begin{equation}\label{eq:d_DS_GV}
\delta_{DS,GV}(R,\rho)=\min \left\{ d_{GV}(R), \min_{\iota} \iota+H^{-1}\left(1-H\left({\iota+\iota\log_2(3)\over 1-R+\rho}\right)\right)\right\}.
\end{equation}

\end{theorem}
A proof can be found in Appendix~\ref{app:dGV}.

It is instructive to compare the bounds $\delta_{DS,GV}(R,\rho)$ and $\delta_{GV}(R)$. In the left part of Fig.~\ref{fig:gv_bound}, we plot these bounds for the case $\rho=0$.
One can see that $\delta_{DS,GV}(R,0)<\delta_{GV}(R)$, especially for low rate quantum codes. This  means that DS codes with $\rho=0$ have inferior performance compared to stabilizer codes (in which only qubits are valnurable  to errors). However, we can
improve DS codes by taking nonzero $\rho$.
It is not difficult to see that $\delta_{DS,GV}(R,\rho)$ grows with $\rho$. So, for each $R$ we can choose $\rho^*(R)$ so that $\delta_{DS,GV}(R,\rho^*(R))=\delta_{GV}(R)$. It happened that  $\rho^*(R)<1-R$ for any $R$ (that is the corresponding $r^*(R)=\rho^*(R)n<n-k$, what we assumed for ensemble $\cE_{n,k,r}$).
In the right part of Fig.~\ref{fig:gv_bound} we plot the normalized length of syndrome $\mu=m/n=1-R$ for stabilizer codes and $\mu+\rho^*(R)=1-R+\rho^*(R)$ for DS codes. One can observe that $\rho^*(R)$ is not very large even for low rate quantum codes. This means that relatively small number of additional generator measurements are needed for achieving the quantum GV bound by DS codes.
\begin{figure}[h]
%\vspace{-2.0cm}
%\includegraphics[width=7.5cm]{GV_DSGV_with_rho0.pdf}
%\includegraphics[width=7.5cm]{needed_rho1.pdf}
%\vspace{-2.0cm}
\[
\includegraphics[width=15cm]{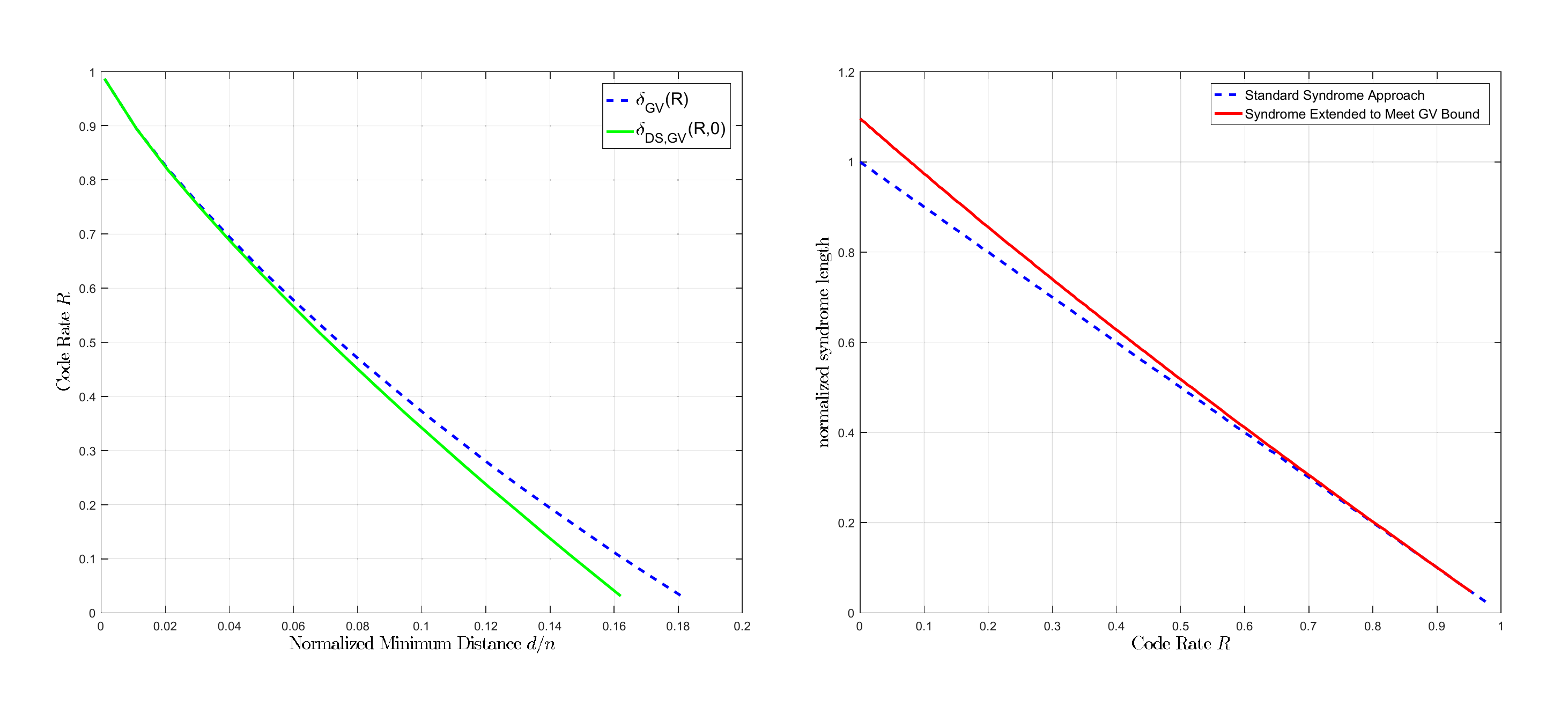}
\]
%\cnote{I don't have the figures}
\caption{Left Figure: Gilbert-Varshamov bounds for stabilizer and DS codes with $\rho=0$.
Right Figure: normalized number of syndrome bits  $(n-k)/n$ for standard stabilizer codes and extended syndrome bits $(n-k+r)/n$ for DS codes, where $r$ is chosen such that  $\delta_{GV,DS}(R,r/n)=\delta_{GV}(R)$.}  \label{fig:gv_bound}
\end{figure}
%\begin{figure}[h]\label{fig:Hamming bound} %\label{fig:GV\tDSGV_rho0}
%\includegraphics[width=7.5cm]{needed_rho1.pdf}
%\vspace{-1.0cm}
%\caption{.
% } \label{fig:Hamming bound} %\label{fig:distillation_repetition code}
%%\vspace{-0.2cm}
%\end{figure}

%Since $B_{i,j}(C^\perp)$ are integers, for sufficiently large $n$ we get that  $B_{\lfloor \iota n\rfloor,\lfloor \xi n\rfloor}^\perp=0$ as soon as $b_{\iota,\xi}^\perp<0$. Summarizing this we get the following result.
%\begin{theorem}
%For a given $\xi$ let $\iota^*$ be the root of
%$$
%H(\iota)+\iota \log_2(3)+RH(\xi R)-R=0.
%$$
%Then there exist DS codes with
%$
%B_{\lfloor \iota n\rfloor,\lfloor \xi n\rfloor}^\perp=0, \iota<\iota^*.
%$
%\end{theorem}
%This Theorem can be used for getting an estimate on the decoding fidelity of a random DS code. This estimate will be obtained in another work.
%
%

\section{CSS-type Quantum DS Codes}\label{sec:CSS}

%\anote{Ching-Yi some parts of this can be improved. In particular it is not clear whether proofs of the Theorems and Corollary follow from the text preceding them? If so, please add something like "From the above arguments the following result follows." Overall please read this Section again and make changes if needed.}
%\cnote{I have uncommented some proofs and added more explanation. However, the details of QR codes are not inculded.}
%\anote{Some material in Section is commented out. I believe we did this because we had to shorten the paper for ISIT. Now we can present this material with all needed details. Please consider this possibility.}

In this section we discuss CSS-type DS codes with $r\geq 0$. Suppose that  $
H_{\mathrm{CSS}}=\left(\begin{array}{c|c}{H}&  0 \\
0 &{H} \end{array}\right)
$ defines  an $[[n,k,d]]$ CSS code, %defined by the matrix
%\begin{equation}\label{eq:H_CSS}
where $H$ is a binary $(\frac{n-k}{2}) \times n$ matrix and $HH^T={\bf 0}$.  
Let
\begin{align}
H\ds=\left(\begin{array}{cc|cc}{H'}& {I}_{(m+r)/2}  & 0 & 0\\0 & 0 &{H'} &{I}_{(m+r)/2} \end{array}\right), \label{eq:modified check matrix}
\end{align}
where
%\begin{equation}\label{eq:H'_CSS}
$
H'^T=\left(\begin{array}{cccc}
H^T&
\f_1^T&
\cdots&
\f_{r/2}^T
\end{array}\right),
$
%\end{equation}
and vectors $\f_j$ are obtained as linear combinations of rows of $H$.
The matrix $\tilde{H}=[H'\ {I_{(m+r)/2}}]$ defines a classical $[n',k',d']$ code.
The minimum distance of the  corresponding DS code is $d'\leq d$ and therefore we obtain an $[[n,k,d': r ]]$ quantum DS code.
 Below we discuss how to extend $H$ to $H'$ so that the minimum distance of the DS code would not decrease and remain equal to $d$.
 
%If $\tilde{H}\bv =0$, then the minimum distance of the DS code is upper bounded by $\wt(\bv)$.
% It is easy to see that if $\tilde{H}$ does not contain identical columns and columns of weight $1$ then $\wt(\bv)\ge 3$ and therefore the minimum distance of the DS code is at least $3$.

For a vector $\by\in \mathbb{F}_2^{n+(m+r)/2}$ we define the extended syndrome as $\bs=(s_1,\ldots,s_{(m+r)/2})=H'\by$. 
One can see that these syndromes belong to the column space of $H'$.
Hence if any nonzero vector $\bw$ from the column space of $H'$ has weight $\wt(\bw)\ge d$, then for any two extended syndromes, say $\bs$ and $\bs'$, we have $\dist(\bs,\bs')\ge d$ and hence the DS code can correct any  $\lfloor \frac{d-1}{2}\rfloor$ syndrome bit errors. If the CSS code defined by $H_{\mathrm{CSS}}$  also has minimum distance $d$ or larger  then the DS code can correct any combination of qubit and syndrome errors whose total number does not exceed
 $\lfloor \frac{d-1}{2}\rfloor$. This leads us to the following result. 

\begin{theorem} \label{thm:cyclic code}
If there exists an $[n,k,d]$ classical dual-containing  cyclic code $C$ with $2k>n$,
then there exists an $[[n,2k-n,d: r]]$ quantum DS code with $r\leq 2k$. %, given $2k>n$.
\end{theorem}
\proof
	Suppose that $H$ is an $n\times (n-k)$ parity-check matrix of $C$. Since $C$ contains its dual code $C^\perp$, we have that $HH^T=0$. Hence $H$ can be used for construction of a CSS code according to $H_{\mathrm{CSS}}$. Let $\bc=(c_0\ c_1\ \cdots\ c_{n-1})\in C^\perp$. Since $C^\perp$ is also cyclic, any cyclic shift of $\bc$ is also a codeword of $C^\perp$. Hence  $n-k$ cyclic shifts of $\bc$ can be used to construct $H$ and additional $k$ cyclic shifts can be used to construct $H'$.   Thus, we can construct $H'$ as follows
\begin{align*}
H'=\begin{pmatrix}
c_0& c_1 & \cdots & c_{n-1}\\
c_1& c_2 & \cdots & c_{0}\\
\vdots&\vdots &\ddots &\vdots \\
c_{n-1}& c_0 & \cdots & c_{n-2}\\
\end{pmatrix}.
\end{align*}
Clearly, the column space and the row spaces of $H'$ are the same and they generate code $C^\perp$. Since $C^\perp \subseteq C$ we have $d(C^\perp)\ge d$.
Now from the  arguments preceding this theorem, it follows that $H'$ defines an $[[n,2k-n,d: 2k]]$ DS code.
% since $2k>n$.
\QED

To demonstrate an application of the above theorem, we consider quadratic-residue (QR) codes.
QR codes are cyclic codes and they are dual-containing for certain parameters~\cite{MS77,LL11}, and therefore can be used for construction of CSS quantum codes.
In particular, they lead to $[[p,1,d]]$   CSS codes with $d^2 - d + 1 \geq p$ for $p=8j-1$.
%Applying Theorem~\ref{thm:cyclic code},  we have the following result.
\begin{theorem} \label{thm:QR_double_even}
For $p=8j-1$,
there exist  $[[p,1,d]]$  CSS codes with $d^2 - d + 1 \geq p$.
Moreover, there are $[[p,1,d:r]]$ quantum DS codes with $r\leq p+1$.
\end{theorem}
\begin{example}
Consider the QR code with $p=23$.
Suppose $H'$ is cyclicly generated by $11+r$ cyclic shifts of a code vector of the dual code.
Table \ref{tb:QR23} provides the distances of the corresponding DS codes with different values of $r$.
%According to Theorem \ref{thm:cyclic code}, we need at most $24$ redundant stabilizers to have a quantum DS code with distance $7$.
Table \ref{tb:QR23} shows that there exists $[[23,1,7:18]]$ quantum DS code and therefore we need only $18$ additional redundant stabilizers, instead of $24$ by Theorem\,\ref{thm:cyclic code}.
\begin{table}[h]
\[
\begin{tabular}{|c|c|c|c|c|c|c|c|c|c|c|c|}
  \hline
  $r$ & 0 & 1 &2 & 3 & 4 & 5 & 6 & 7 & 8 & 9 \\
\hline
  $d$ & 3 & 4 &4 & 4 & 5 & 5 & 5 & 6 & 6 & 7 \\
 %  &  &  &  &  &  &  &  &  &  \\
 %  &  &  &  &  &  &  &  &  &  \\
 %  &  &  &  &  &  &  &  &  &  \\
  \hline
\end{tabular}
\]  \caption{The distance of $\tilde{H}$ corresponding to  different number of rows. } \label{tb:QR23}
\vspace{-0.7cm}
\end{table}
\end{example}

The family of quantum QR codes in Theorem \ref{thm:QR_double_even}, which includes the Steane code and the quantum Golay code, are
important in the theory of fault-tolerant quantum computation. In particular they are used for finding error thresholds \cite{PR12}.
Here we have shown that these codes also induce nontrivial quantum DS codes.

\section*{ACKNOWLEDGEMENT}
CYL was was financially supported from the Young Scholar Fellowship Program by Ministry of Science and Technology (MOST) in Taiwan, under
Grant MOST107-2636-E-009-005.
%\cnote{Please add your grants or email me.}

\appendix %\label{sec:appendix}
%\cnote{Appendix is divided into subsections for clarity}
%\section{Appendix A}

\subsection{Properties of Krawchuk polynomial}\label{app:Kraw}
The following equalities holds (see \cite[Chapter 5]{MS77})
\begin{align}
&K_0(x;n,q)=1, \label{eq:K0=1}\\
&K_j(0;n,q)=(q-1)^j {n \choose j}, \label{eq:K_j(0)}\\
&\sum_{i=0}^n K_r(i;n,c)K_i(s;n,c)=c^n\delta_{rs},c=2\mbox{ or }4, \label{eq:orth_Kraw}\\
&\sum_{j=0}^n {n\choose j}(q-1)^j K_i(j;n,q)=q^n \delta_{i,0},  \label{eq:(n,j)K_i(x)}\\
&\sum_{i=0}^n {n-i \choose n-j}K_i(x;n,q)=q^j{ n-x \choose j}. \label{eq:(n-i,n-j)K_i(x)}
\end{align}
In \cite[eq. A.19]{MRRW77} and \cite[Lemma 2]{AL99}, it is shown that
\begin{align}
K_a(j;m,2)K_b(j;m,2)&=\sum_{u=0}^m \beta(u,a,b)K_u(j;m,2),\label{eq:KaKb_sumKu2}\\
K_g(i; n,4)K_h(i; n,4)&=\sum_{q=0}^n \sum_{w=0}^{n-q} \alpha(q,g,h,w) K_q(i; n,4),\label{eq:KgKh_sumKq4}
\end{align}
where
$\beta(u,a,b)$ and $\alpha(q,g,h,w)$ are defined in (\ref{eq:beta}) and (\ref{eq:alpha}) respectively.
%Krawtchouk polynomials  are orthogonal with respect to the following product:
%\begin{align}
%\sum_{i=0}^n K_r(i;n,c)K_i(s;n,c)=c^n\delta_{rs},c=2\mbox{ or }4. \label{eq:orth_Kraw}
%\end{align}
% \bibliographystyle{IEEEtran}
% \bibliography{IEEEabrv,../qecc}
%For our purposes we need the following property.
\begin{lemma}\label{lem:sum_binomial_times_Krawtchouk}
	\begin{equation}\label{eq:sum_binomial_times_Krawtchouk}
	\sum_{j=0}^m {m\choose j} K_u(j;n,2)=2^m{n-m \choose u}.
	\end{equation}
\end{lemma}
\proof
The generating function of the binary Krawtchouk
polynomials (see \cite[Sec.\,5.7]{MS77}) is
$
(1+x)^{n-j}(1-x)^j=\sum_{u=0}^n K_u(j;n,2)x^u.
$
Using this equauion, we obuain
$
\sum_{j=0}^m {m\choose j} (1+x)^{n-j}(1-x)^j=\sum_{u=0}^n x^u \sum_{j=0}^m {m\choose j} K_u(j;n,2).
$
At the same time
\begin{align*}
&\sum_{j=0}^m {m\choose j} (1+x)^{n-j}(1-x)^j=\sum_{j=0}^m {m\choose j} (1+x)^{n-m}(1+x)^{m-j}(1-x)^j\\
=&(1+x)^{n-m}\sum_{j=0}^m {m\choose j} (1+x)^{m-j}(1-x)^j=(1+x)^{n-m} 2^m=2^m \sum_{u=0}^{n-m} {n-m\choose u}x^u.
\end{align*}
Comparing these two expressions, we finish the proof. \hfill\QED

\subsection{Proof of Theorem~\ref{thm:MaWil}}\label{app:MaWil}
\proof
	We can  use the techniques in~\cite{LHL14} as follows. We define a Fourier transform operator with respect to the inner product (\ref{eq:inner_prod}) and find a MacWilliams identity that relates the two split weight enumerators. Then Theorem~\ref{thm:MaWil} follows directly.
	% by equating the coefficients on both sides of the identity.
%}

%\vspace{0.2cm}
%\subsection{Proof of Theorem~\ref{thm:SingletonBounNonDeg}}\label{app:singletonNonDeg}
%Suppose that code $C\dsp$ has minimum distance $d$. Let us puncture $d-1$ coordinates among the $n$ coordinates over $\mathbb{F}_4$. Then we will have a code of size $2^{2n}$ in which codewords are vectors of the form $\bv=(v_1, \ldots, v_{n-d+1}, w_1, \ldots w_{n-k})$ with $v_j\in \mathbb{F}_4$ and $w_i \in \mathbb{F}_2$. The total number of such vectors is $2^{2(n-d+1)+n-k}$ and this number should not exceed the code size $2^{2n}$. From this, the assertion follows. \QED

%\vspace{0.2cm}
\subsection{Proof of Lemma \ref{lem:Ham_bound}}\label{app:Hamming}

Using (\ref{eq:KaKb_sumKu2}) and (\ref{eq:KgKh_sumKq4}), we obtain
\begin{align*}
f_{i,j}^{(k)}=&\sum_{a=0}^t\sum_{b=0}^t K_a(j;m,2) K_b(j;m,2) \sum_{g=0}^{t-\lambda a}  \sum_{h=0}^{t-\lambda b} K_g(i;n,4)K_h(i;n,4)\\
=&
\sum_{a=0}^t\sum_{b=0}^t \sum_{u=0}^m \beta(u,a,b) K_u(j;m,2)
+ \sum_{g=0}^{t-\lambda a} \sum_{h=0}^{t-\lambda b} \sum_{q=0}^n \sum_{w=0}^{n-q}
\alpha(q,g,h,w) K_q(i; n,4).
\end{align*}
Now, using (\ref{eq:KrawExpension}), we obtain
\begin{align*}
&f^{(k)}(l,r)=\sum_{i=0}^n \sum_{j=0}^m   f_{i,j}^{(k)} K_i(l;n,4) K_j(r;m,2)\\
 =&\sum_{i=0}^n \sum_{j=0}^m \left[\sum_{a=0}^t\sum_{b=0}^t \sum_{u=0}^m \beta(u,a,b) K_u(j;m,2) K_j(r;m,2)\right.+\left.\sum_{g=0}^n \sum_{h=0}^{t-\lambda a} \sum_{q=0}^{t-\lambda b} \sum_{w=0}^{n-q} \alpha(q,g,h,w)  K_q(i; n,4) K_i(l;n,4)  \right]\\
=&
\sum_{a=0}^t\sum_{b=0}^t \sum_{u=0}^m  \beta(u,a,b) \sum_{j=0}^m K_u(j;m,2) K_j(r;m,2) +\sum_{g=0}^{t-\lambda a} \sum_{h=0}^{t-\lambda b} \sum_{q=0}^n \sum_{w=0}^{n-q} \alpha(q,g,h,w) \sum_{i=0}^n  K_q(i; n,4) K_i(l;n,4) \\
=& 4^n 2^m \sum_{a=0}^t\sum_{b=0}^t \beta(r,a,b) \sum_{g=0}^{t-\lambda a} \sum_{h=0}^{t-\lambda b} \sum_{w=0}^{l-q} \alpha(l,g,h,w),
\end{align*}
where in the last step we used the orthogonality property of Krawtchouk polynomials (\ref{eq:orth_Kraw}).
%\begin{align}
%\sum_{i=0}^n K_r(i;n,c)K_i(s;n,c)=c^n\delta_{rs},c=2\mbox{ or }4. \label{eq:orth_Kraw}
%\end{align}
\hfill\QED

\subsection{Proof of~Lemma \ref{lem:f(x,y)=0}} \label{app:f(x,y)=0}

A binomial coefficient ${i\choose j}$ is assumed to be zero if:
% at least one of the following conditions holds:
 1) $i<j$, 2) $j<0$, or 3) $j$ is not an integer.
The polynomial $f^{(k)}(x,y)$ is a sum of non negative terms:
$$
{m-y\choose (a+b-y)/2}{y\choose (a-b+y)/2}{x\choose 2x+2w-g-h}{n-x\choose w}{2x+2w-g-h \choose x+w-h} 2^{g+h-2w-q}3^w.
$$
A particular term is not zero if all the five binomial coefficients are not zeros. Each of those binomial coefficients is not zero if and only if neither of the above conditions hold, e.g, the first binomial coefficient is not zero if and only if
$a+b\ge y,~m-y\ge (a+b-y)/2,\mbox{ and } (a+b-y)/2 \mbox{ is an integer}.$
In the following discussion, we drop condition 3 since it is not needed for our purpose. From the above arguments it follows that $f^{(k)}(x,y)>0$ for $x+y\ge 2t+1$ if only if there is a solution to the system of linear inequalities
$A (a, b, w, g, h, x, y)^T\le {\bf b},$
with
{\small
\begin{align*}
A^T=&\left(\begin{array}{rrrrrrrrrrrrrrrrrrr}
-1 & 1 & -1 & 1 &   &   &   &   &  -1 & 1  &   &   &   &  1 &   &   &   &   &   \\
-1 & 1 & 1  & -1&   &   &   &   &     &    & -1& 1 &   &    &   & 1 &   &   &   \\
&   &    &   & -2& 2 & -1& -1&     &    &   &   &   &    &   &   & -1& 1 &   \\
&   &    &   & 1 & -1&   & 1 &     &    &   &   & -1&  1 &   &   &   &   &   \\
&   &    &   & 1 & -1&  1&   &     &    &   &   &   &    & -1& 1 &   &   &   \\
&   &    &   & -2& 1 & -1& -1&     &    &   &   &   &    &   &   &   & 1 & -1\\
1 & 1 & -1 & -1&   &   &   &   &     &    &   &   &   &    &   &   &   &   & -1
\end{array}\right), \\
%\mbox{and}&\\
{\bf b}=&\left(\begin{array}{rrrrrrrrrrrrrrrrrrr}
\phantom{....} & 2m  & \phantom{-1}  &  \phantom{-1}  &   \phantom{-2}  &  \phantom{-1}   &  \phantom{-1}   &  \phantom{-1}   &   \phantom{-1}    &  t &   \phantom{-1}  & t & \phantom{-1}  & t  &  \phantom{-1}   &  t&  \phantom{-1}   & n  & -2t-1
\end{array}\right).
\end{align*}
}
Conducting the Fourier--Motzkin elimination~\cite{Sch98} in the order of $h, y, g, a, b, w, x$ (any other order can  also  be used, but this one requires  computations that are not too long),
we come to the incompatible condition $0\le -1/2$. This completes the proof. \hfill\QED
\subsection{Proof of Theorem \ref{thm:aver_weight_distribution}} \label{app:avg}
We would like to analyze the weight distribution of a random DS code from $\cE_{n,k,r}$ with a generator matrix of the form
$\left[\begin{array}{ccc}
H & I_{m} & 0 \\
0 & A & I_r
\end{array}\right],
$
where $m=n-k$.
Let $\cE_{n,m}$ be the set of DS codes with a generator matrix of the form   $[H \ I_{m}]$, and $\cF_{m,r}$ be the set of binary codes with a generator matrix of the form $[A\ I_r]$, where $A$ has rank $r$.
A code from $\cE_{n,k,r}$ can be considered as a combination of codes from $\cE_{n,m}$  and  $\cF_{m,r}$.
%The number of distinct codes in $\cE_{n,k,r}$ can be calculated from the combinations of codes from  $\cE_{n,m}$  and  $\cF_{m,r}$.
%Thus we first prove the following lemmata regarding the sizes of $\cE_{n,m}$  and  $\cF_{m,r}$.
\begin{lemma}\label{lem:Number of Snkr codes} The size  of the ensemble  $\cE_{n,m}$  is
	\begin{align}
	|\cE_{n,m}|=|\cE_{n,m}^\perp|=\prod_{u=0}^{m-1} {(2^{2(n-u)}-1)(2^{m}-2^u)\over 2^{u+1}-1},
	\end{align}
	and any vector $\bw=(\ba,\bb)$ with $\ba\in \mF_4^n\setminus {\bf 0}$ and $\bb\in \mF_2^{m}\setminus {\bf 0}$ is contained
	in
	\begin{align}
	L=\prod_{u=1}^{m-1} {(2^{2(n-u)}-1)(2^{m}-2^u)\over 2^{u}-1}
	\end{align}
	codes from $\cE_{n,m}$.
\end{lemma}
\proof
It is proved that the number of $[n,m]$ additive self-orthogonal codes over $\mathbb{F}_4$
is
$
S\triangleq\prod_{u=0}^{m-1} {(2^{2(n-u)}-1)\over 2^{u+1}-1}
$ \cite{A14}.
For any $[n,m]$ additive self-orthogonal code, we can choose $m$ generators (rows of matrix $H$) in
$
T\triangleq\prod_{u=0}^{m-1} (2^{ m}-2^u)
$
ways. Hence, using any $[n,m]$ self-orthogonal code, we can form $T$ different matrices $[H\ I_{m}]$. Thus,
$|\cE_{n,m}|=ST$. %which gives the first statement.
%\cnote{Isn't $T=  \prod_{u=0}^{n-k-1} (2^{ (n-k)}-2^u)$? For an additive code with $n-k$ generators, the number of different codewords is $2^{n-k}$?}
%\cnote{Fixed a typo in $T$  and the statement remains the same.}

It is shown  that any nonzero vector $\ba\in \mF_4^n\setminus{\bf 0}$  is contained in
$P=\prod_{u=1}^{m-1} {2^{2(n-u)}-1\over 2^{u}-1}$~~
$[n,m]$ self-orthogonal codes \cite{A14}. We can use any of those codes for building a code from $\cE_{n,m}$ with vector $(\ba,\bb)$ as its first basis vector. The other $(m-1)$ basis vectors  can be
chosen in
$
R=\prod_{u=1}^{m-1} (2^{m}-2^u)
$
ways. Hence any $(\ba,\bb)$ is contained in $PR$ codes from $\cE_{n,m}$. %and this leads to the second statement.
\hfill\QED

%\cnote{I put some results in the following lemma for clarification.}
\begin{lemma}  \label{lemma:Fmr}The size  of the ensemble  $\cF_{m,r}$  is
	\begin{align}
	|\cF_{m,r}|=\left[{m\atop r}\right]\prod_{u=0}^{r-1} (2^r-2^u)=\prod_{u=0}^{r-1} (2^m-2^u),
	\end{align}
	where
$\left[{m\atop r}\right]=\frac{(2^m-1)(2^{m-1}-1)\cdots (2^{m-r+1}-1)}{(2^r-1)(2^{r-1}-1)\cdots (2-1)}$
 is the \emph{binary Gaussian binomial coefficient}, and any vector $(\bb,\bc)$ for $\bb\in \mF_2^{m}\setminus {\bf 0}$ and $\bc\in \mF_2^{r}\setminus {\bf 0}$ is contained in \begin{align}
	\left[{m-1\atop r-1}\right]\prod_{u=1}^{r-1} (2^r-2^u)=\prod_{u=1}^{r-1} (2^m-2^u)\end{align} codes from $\cF_{m,r}$.
\end{lemma}
The proof of this lemma is similar to the previous one and is omitted.
%\proof
%An $[m,r]$ binary linear code from $\cF_{m,r}$ is defined by a full-rank matrix $A$, which can be constructed in $\left[{m\atop r}\right]$ ways (see \cite[Ch.\,15]{MS77}).
%For any such code, we may choose  $r$ generators (rows of matrix $A$) in $ \prod_{u=0}^{r-1} (2^r-2^u)$ ways
%to construct a generator matrix $[A\ I_r]$. Thus we have the first statement.
%%\cnote{What is  $\left[{m\atop r}\right]$? Is it  $\left({m\atop r}\right)$?}
%The second statement follows along the same line, except that we inlude a given vector as the default first generator. \hfill\QED
%There are $\left(2^m-2\right)\left(2^m-2^2\right)\cdots \left(2^m-2^{r-1}\right)$ ways to find a basis of $r$ vectors over $\mF_2^m$ with $\bb$ as the first basis vector.
%Among them
%Any nonzero vector $\bb \in \mF_2^m\setminus{\bf 0}$  is contained in
%$\left[{m-1\atop r-1}\right]$~~
%$[m,r]$ additive self-orthogonal codes. We can use any of those codes for building a code from $\cF_{m,r}$ with vector $(\ba,\bb)$ as its first basis vector. The other $(n-k-1)$ basis vectors  can be
%chosen in
%$
%R=\prod_{u=1}^{n-k-1} (2^{n-k}-2^u)
%$
%ways.

%Now we use these lemmata for the following computations to determine the number of codes from $\cE_{n,k,r}$ that contain a vector $(\ba, \bb,{\bc})$ for $\ba\in \mF_2^n, \bb\in \mF_2^m,\bc\in \mF_2^r$.

\begin{lemma}
(1) Any vector $(\ba, \bb,\bc)$, where  $\ba\in \mF_4^n,\setminus{\bf 0}$, $\bb\in\mF_2^{m},\setminus{\bf 0}$, $\bc \in \mF_2^r\setminus{\bf 0}$, is contained in
\begin{equation}\label{eq:abc}
F(\ba, \bb,\bc)=L(2^m-2)\prod_{u=1}^{r-1}(2^m-2^u)
\end{equation}
codes from $\cE_{n,k,r}$.
(2) Any  $(\ba, \bb,{\bf 0})$, where $\ba\in \mF_4^n\setminus{\bf 0}$, $\bb\in\mF_2^{m}\setminus{\bf 0}$, ${\bf 0}\in \mF_2^r$, is contained in
\begin{equation}\label{eq:ab0}
F(\ba, \bb,{\bf 0})=L\prod_{u=0}^{r-1} (2^m-2^u)
\end{equation}
codes from $\cE_{n,k,r}$. 
%Indeed, a vector $(\ba, \bb,{\bf 0})$  can be obtained only as the sum of a vector $(\ba,\bb,{\bf 0})$
%(where $(\ba,\bb)$ is a code vector of a code from $\cE_{n,m}$) and  code vector ${\bf 0}$ of a code from $\cF_{m,r}$.
%A vector $(\ba,\bb)$ is contained in $L$ codes from $\cE_{n,m}$, and ${\bf 0}$ is contained in all codes from $\cF_{m,r}$.
%The product $L |\cF_{m,r}|$  leads to (\ref{eq:ab0}).
(3) Any  $({\bf 0}, \bb,\bc)$,  where ${\bf 0}\in \mF_4^n$, $\bb\in\mF_2^{m}\setminus{\bf 0}$, $\bc\in \mF_2^r\setminus{\bf 0}$,  is contained in
\begin{equation}\label{eq:0bc}
F({\bf 0}, \bb,\bc)=ST \prod_{u=1}^{r-1} (2^m-2^u)
\end{equation}
codes from $\cE_{n,k,r}$. 
%A vector $({\bf 0}, \bb,\bc)$ can be obtained only as the sum of a code vector $({\bf 0}_{n+m},{\bf 0}_r)$ (where $({\bf 0}_{n+m}$ belongs to
%a code from $\cE_{n,m}$ and ${\bf 0}_r\in \mF_2^r$)
%and a vector $({\bf 0},\bb,\bc)$ (where $(\bb,\bc)$ is a code vector of a code from $\cF_{m,r}$). Since any given $(\bb,\bc)$ is contained in $ \prod_{u=1}^{r-1} (2^m-2^u)$ codes from $\cF_{m,r}$ by Lemma~\ref{lemma:Fmr} and $|\cE_{n,m}|=ST$ by Lemma~\ref{lem:Number of Snkr codes}, we obtain (\ref{eq:0bc}).
%\cnote{Replace $S$ in (\ref{eq:0bc}) by $ST$ since $|\cE_{n,m}|=ST$  }
(4) Any   $(\ba,{\bf 0},\bc)$,  where  $\ba\in \mF_4^n,\setminus{\bf 0}$, ${\bf 0}\in\mF_2^{m}$, $\bc \in \mF_2^r\setminus{\bf 0}$, is contained in
\begin{equation}\label{eq:a0c}
F(\ba,{\bf 0},\bc)=L\prod_{u=0}^{r-1} (2^m-2^u)
\end{equation}
codes from $\cE_{n,k,r}$.

\end{lemma}
\begin{proof}
	We prove the first one and the other three follow similarly. 
	
	A vector $(\ba, \bb,\bc)$ can be obtained only as the sum of a vector $(\ba,\x,{\bf 0})$ (where $\x\in \mF_2^m\setminus{\bf 0},\x\not=\bb,~{\bf 0}\in \mF_2^r$, and $(\ba,\x)$ is a code vector of a code from $\cE_{n,m}$)
	and a code vector $({\bf 0},\bb+\x,\bc)$ (where $(\bb+\x,\bc)$ is a code vector of of a code from $\cF_{m,r}$). Any given $(\bb+\x,\bc)$ is contained in $ \prod_{u=1}^{r-1} (2^m-2^u)$ codes from $\cF_{m,r}$ by Lemma~\ref{lemma:Fmr}. Since $(\ba,\x)$
	is contained in $L$ codes from $\cE_{n,m}$ and
	vector $\x$ can be chosen in $2^m-2$ ways, we have $(\ba, \bb,\bc)$ contained in $L(2^m-2) \prod_{u=1}^{r-1} (2^m-2^u)$ codes from $\cE_{n,k,r}$, which gives  (\ref{eq:abc}).
\end{proof}

%We again observe that $(\ba,{\bf 0},\bc)$ should be obtained as the sum of a vector $(\ba,\bb,{\bf 0})$ (where $(\ba,\bb)$ belongs to a code from $\cE_{n,m}$) and the vector $({\bf 0},\bb,\bc)$ (where
%$(\bb,\bc)$ belongs to a code from $\cF_{m,r}$). In addition, we have $(2^m-1)$ choices for $\bb$. Computations similar to the ones used before lead us to (\ref{eq:a0c}).
%\cnote{Fix (\ref{eq:a0c})  by replacing $F(\ba,{\bf 0},\bc)=L\prod_{u=0}^{r-1} (2^m-2^u)$ with $F(\ba,{\bf 0},\bc)=L\prod_{u=1}^{r-1} (2^m-2^u)$  according to Lemma~\ref{lemma:Fmr} }

In addition, we note that the total number of codes in $\cE_{n,k,r}$ is
\begin{align}
N=|S_{n,k,r}|=ST \left[{m\atop r}\right] \prod_{u=0}^{r-1} (2^r-2^u)= ST \prod_{u=0}^{r-1} (2^m-2^u).
\end{align}
%\cnote{Replaced $S$ by $ST$}

Now we have all the ingredients needed for finding $\overline{B}_{i,j}$. It will be convenient to assume that ${a\choose b}=0$ if $a<b$ or $b<0$. Let consider the case of $i>0$ and $j>0$. Then
\begin{align*}
\overline{B}_{i,j}%&={1\over N} \sum_{C\in \cE_{n,k,r}} B_{i,j}(C) \\
=&{1\over N} \left(\sum_{{(\ba,\bb,{\bf 0}):\atop \mbox{\tiny wt}(\ba)=i,\mbox{\tiny wt}(\bb)=j}} F(\ba,\bb,{\bf 0}) + \sum_{{(\ba,{\bf 0},\bc):\atop \mbox{\tiny wt}(\ba)=i,\mbox{\tiny wt}(\bc)=j}} F(\ba,{\bf 0},\bc) +
\sum_{{(\ba,\bb,\bc):\atop \mbox{\tiny wt}(\ba)=i,\mbox{\tiny wt}(\bb)+\mbox{\tiny wt}(\bc)=j}} F(\ba,\bb,\bc)\right)\\
=&{1\over N} {n\choose i} 3^i\left( {m\choose j} F(\ba,\bb,{\bf 0}) + {r\choose j} F(\ba,{\bf 0},\bc) + F(\ba,\bb,\bc) \sum_{u=1}^{j-1} {m\choose u}{r\choose j-u} \right).
\end{align*}
%\cnote{I revised  the last term since $\bb\neq 0, \bc\neq 0$}
Taking into account that $\sum_{u=0}^{j} {m\choose u}{r\choose j-u}={m+r \choose j}$, after some computations, we obtain (\ref{eq:Bij}). Equation\,(\ref{eq:B0j}) is obtained in a similar way.
%\cnote{the formula should be $\sum_{u=0}^{j} {m\choose u}{r\choose j-u}={m+r \choose j}$,
%	instead of $\sum_{u=0}^{m} {m\choose u}{r\choose j-u}={m+r \choose j}$}

To derive $\overline{B}_{i,j}^\perp$ we  use MacWilliams identities (\ref{eq:MacW}), which also hold for average weight enumerators $\overline{B}_{i,j}$ and $\overline{B}_{i,j}^\perp$. Changing the role of codes $C$ and $C^\perp$, we get, similar to (\ref{eq:MacW}):
\begin{align*}
\overline{B}_{i,j}^\perp%=&{1\over 2^{m+r} }\sum_{l=0}^n\sum_{t=0}^{m+r} \overline{B}_{l,t}^\perp K_i(l;n,4) K_j(t;m+r,2)\\
=&{1\over 2^{m+r} (4^n - 1)(2^m - 1)}\sum_{l=1}^n {n\choose l} l^3 K_i(l;n,4) \times  \sum_{t=1}^{m+r}  ((2^m - 2)
{r + m\choose t}+{m \choose t}+{r\choose t}) K_j(t;m+r,2) \\
&+  {1\over 2^{m+r} }  K_i(0;n,4)  \sum_{t=1}^{m+r} \left({m + r\choose t} -{m\choose t}-{r\choose t}\right) K_j(t;m+r,2)+{1\over 2^{m+r} }K_i(0;n,4) K_j(0;m+r,2).
\end{align*}

Using (\ref{eq:K_j(0)}), (\ref{eq:(n,j)K_i(x)}), and (\ref{eq:sum_binomial_times_Krawtchouk}), after long manipulations, we obtain (\ref{eq:Bd_i0}), (\ref{eq:Bd_ij}), and (\ref{eq:Bd_0j}).

\hfill\QED

\subsection{Proof of Theorem \ref{thm:asympt_aver_weight_distribution}} \label{app:dGV}
According to Markov's inequality for a given pair $i$ and $j$, we have
$$
\Pr\left(B_{i,j}(C^\perp)\ge  ((n+1)(m+r+1))^{1+\epsilon} \overline{B}_{i,j}^\perp \right)\le {1\over ((n+1)(m+r+1))^{1+\epsilon}},
$$
for any $\epsilon>0$.
Applying the union bound, we obtain
$$
\Pr(B_{i,j}(C^\perp)\ge ((n+1)(m+r+1))^{1+\epsilon} \overline{B}_{i,j}^\perp \mbox{ for at least one pair } i,j)\le {1\over ((n+1)(m+r+1))^{\epsilon}},
$$
and further
$
\Pr\left(B_{i,j}(C^\perp)< ((n+1)(m+r+1))^{1+\epsilon} \overline{B}_{i,j}^\perp \mbox{ for all } i,j\right)\ge 1-{1\over ((n+1)(m+r+1))^{\epsilon}}.
$
Hence there exists a code $C^\perp\in {\cE}_{n,m}^\perp$ such that
\begin{equation}\label{eq:Bij_expur}
B_{i,j}(C^\perp)\le ((n+1)(m+r+1))^{1+\epsilon} \overline{B}_{i,j}^\perp, \forall i,j.
\end{equation}
Now we  consider codes of growing lengths, i.e., $n\rightarrow \infty$.
Note that $m/n=(n-k)/n=1-R$. Recall that ${1\over n} \log_2 {n \choose i} = H(i/n)+o(1)$ \cite{MS77}. So    the three terms of the last factor of (\ref{eq:Bd_ij})  are that
\begin{align}
{1\over n} \log_2 {m+r \choose j}2^m&=(1-R+\rho)H\left({\xi\over 1-R+\rho}\right)+1-R,\label{eq:term1}\\
 {1\over n} \log_2 {r \choose j}2^m&=\rho H\left({\xi\over \rho}\right)+1-R,\label{eq:term2}\\
 {1\over n} \log_2 {m \choose j}2^r&=(1-R)H\left({\xi\over 1-R}\right)+\rho.\label{eq:term3}
\end{align}
Simple analysis shows that for $\rho\le 1-R$ (which is the same as $r\le n-k$),  we have that (\ref{eq:term1})  is always larger than (\ref{eq:term2}) and (\ref{eq:term3}). Hence
\begin{align*}
\overline{b}_{\iota,\xi}^\perp\triangleq & {1\over n} \log_2 ((n+1)(m+r+1))^{1+\epsilon} \overline{B}_{i,j}^\perp={1\over n} \log_2 \overline{B}_{i,j}^\perp + o(1)\\
=&H(\iota)+\iota\log_2(3)+(1-R+\rho)H\left({\xi\over 1-R+\rho}\right)-(1-R+\rho)+ o(1).
\end{align*}
The equation (\ref{eq:b_ij}) is obtained in a similar way.
%\cnote{it seems to me that $b_{\iota\xi}^\perp$ should be defined according to (\ref{eq:Bij_expur}) and $C^\perp$.
%I don't understand the reasoning for the first and the third inequalities. Does the first one hold up to $o(1)$ since in (\ref{eq:Bij_expur}) we have $\leq$? }
%\cnote{For the third inequality, because  (\ref{eq:term1})  is always larger than (\ref{eq:term2}) and (\ref{eq:term3}),
%only the large term is considered and the two smaller terms are omitted in the asymptotic case?  }
%\anote{I tried to answer those questions in our e-mails. If still there are questions please ask, otherwise please remove those \cnote-s.}

%With these notations, for codes satisfying (\ref{eq:Bij_expur}), we have
%\begin{align}
%b_{\iota,\xi}^\perp&={1\over n}\log_2 B_{\lfloor \iota n\rfloor,\lfloor \xi n\rfloor}^\perp \label{eq:splitGV}\\
%&=H(\iota)+\iota \log_2(3)+RH(\xi R)-R+o(1).\nonumber
%\end{align}

Let us have $C^\perp$ that satisfies (\ref{eq:Bij_expur}).  Since $C^\perp$ is linear, all ${B}_{i,j}(C^\perp)$ are integers. Hence if $\iota^*$ and $\xi^*$ are such that $\overline{b}_{\iota,\xi}\le 0$ for
$\iota\le \iota^*$ and $\xi\le \xi^*$, then $B_{i,j}(C^\perp)=0$ for $1\le i \le (\iota^*-\epsilon)n,~1\le j \le (\xi^*-\omega)n$ for any $\epsilon,\omega>0$ and sufficiently large $n$. It is not difficult to see that if $\iota\le \delta_{GV}(R)$, then
$\overline{b}^\perp_{\iota,0}\le 0$. Similarly, if $\xi(\iota)=H^{-1}\left(1-H\left({\iota+\iota\log_2(3)\over 1-R+\rho}\right)\right)$, then $\overline{b}^\perp_{\iota,\xi(\iota)}=0$. Thus $B_{i,0}(C^\perp)=0$ for all $i\le (\delta_{GV}(R)-
\epsilon)n$ and $B_{i,j}(C^\perp)=0$ if $i+j\le (\iota+\xi(\iota)-\omega)n$. Hence  (\ref{eq:d_DS_GV}) follows. \hfill\QED

\end{document}